%% file: main.tex
\documentclass[conference]{IEEEtran}
\IEEEoverridecommandlockouts
\def\BibTeX{{\rm B\kern-.05em{\sc i\kern-.025em b}\kern-.08em
    T\kern-.1667em\lower.7ex\hbox{E}\kern-.125emX}}

\usepackage[T1]{fontenc}
\usepackage[utf8]{inputenc}
\usepackage{amsmath,amsfonts,amssymb}

\usepackage{amsthm}
\usepackage{bm}
\usepackage{extarrows}
\usepackage{stmaryrd}
\usepackage{proof}
\usepackage{bussproofs}
\usepackage[ruled,vlined,linesnumbered]{algorithm2e}
\usepackage{algpseudocode}
\usepackage[usenames, table, svgnames, dvipsnames]{xcolor}
\usepackage{booktabs}

\usepackage{tabularx}
\usepackage{graphics}

\usepackage[usenames, table, svgnames, dvipsnames]{xcolor}
\usepackage{url}
\usepackage[hyperfootnotes=true,colorlinks,linkcolor=blue,anchorcolor=black,citecolor=blue,urlcolor=black]{hyperref} % tag the hyperref with pure colors

\usepackage{tikz}
\usetikzlibrary{arrows,shapes,snakes,automata,backgrounds,petri,positioning}
\usepackage{diagbox}
\usepackage{makecell}
\newcommand{\oomit}[1]{}

\newcommand{\tttrans}{\rightarrow}
\newcommand{\tr}[1]{\langle #1 \rangle} % For the time restriction 
\newcommand{\semantic}[1]{\left\llbracket #1 \right\rrbracket} % For the semantic of TREs
\newcommand{\cbrackets}[1]{\left\{ #1 \right\}} % For the size of curly brackets
\newcommand{\pbrackets}[1]{\left( #1 \right)} % For the size of parentheses
\newcommand{\abrackets}[1]{\left\langle #1 \right\rangle} % For the size of angle brackets
\newcommand\hole{\heartsuit}

\newtheorem{prob}{Problem}

\usepackage{todonotes}
\newcommand{\aj}[1]{\todo[inline,color=green!10]{\textbf{AJ:} #1}}

\newcommand{\znj}[1]{\todo[inline,color=orange!10]{\textbf{ZNJ:} #1}}

\usepackage{orcidlink}

\newcommand{\myparagraph}[1]{\smallskip\noindent{\bf #1.}}
\allowdisplaybreaks

\spnewtheorem{mytheorem}{Theorem}{\bfseries}{\rmfamily} 
\spnewtheorem{myproblem}{Problem}{\bfseries}{\rmfamily} 
\spnewtheorem{mylemma}{Lemma}{\bfseries}{\rmfamily} 
\spnewtheorem{myexample}{Example}{\bfseries}{\rmfamily}
\spnewtheorem{myassumption}{Assumption}{\bfseries}{\rmfamily}
\spnewtheorem{mydefinition}{Definition}{\bfseries}{\rmfamily}
\spnewtheorem{myremark}{Remark}{\bfseries}{\rmfamily}
\spnewtheorem{proposition}{Proposition}{\bfseries}{\rmfamily}
{\bfseries}{\rmfamily}

\begin{document}
\title{Mining DTA with SMT by Exploiting  Simple Elementary Language and Timed Augmented Prefix Acceptor 
}

%\title{Simplification of SMT-based DTA Mining via Elementary Language\\
%}

% \author{Anonymous Authors}

\author{
\IEEEauthorblockN{
Ziran Wang~\orcidlink{0009-0004-5555-5932}
\hspace{7cm}
Jie An\textsuperscript{*}~\orcidlink{0000-0001-9260-9697}
}
\IEEEauthorblockA{
\begin{minipage}{0.5\linewidth}
\centering
\textit{Key Lab. of System Software (Chinese Academy of Sciences),}\\
\textit{Institute of Software, Chinese Academy of Sciences}\\
\textit{\& University of Chinese Academy of Sciences}\\
Beijing, China\\
wangzr@ios.ac.cn
\end{minipage}
\hfill
\begin{minipage}{0.5\linewidth}
\centering
\textit{National Key Lab. of Space Integrated Information System,}\\
\textit{Institute of Software, Chinese Academy of Sciences}\\
\textit{\& University of Chinese Academy of Sciences}\\
Beijing, China\\
anjie@iscas.ac.cn
\end{minipage}
}

% \newline
\vspace{0.8em}

\IEEEauthorblockN{
Naijun Zhan\textsuperscript{*}~\orcidlink{0000-0003-3298-3817}
}
\IEEEauthorblockA{
\textit{MoE Key Lab. of High Confidence Software Technologies,}\\
\textit{School of Computer Science, Peking University}\\
\textit{\& Zhongguancun Laboratory}\\
Beijing, China\\
njzhan@pku.edu.cn
}
\thanks{* Corresponding authors.}
}

\maketitle

\begin{abstract}
Timed automata, which extend finite state automata by introducing clock variables, serve as a popular formalism for specifying and analyzing the timed behaviors of real-time systems. Extracting the timed behaviors of a black-box, safety-critical system is crucial for designing and analyzing its real-time requirements, yet it remains challenging. In this paper, we address this problem by generating a deterministic timed automaton (DTA) consistent with a given set of system behaviors, comprising both positive and negative examples. To this end, we adapt the formalism of simple elementary languages (sEL) and introduce the timed augmented prefix tree acceptor (tAPTA). Our approach proceeds as follows: First, we preprocess samples by translating them into sEL, which discards redundancy and detects conflicts; then, we rewrite the resulting sELs in an incremental form and construct a tAPTA to further simplify the samples; finally, we encode the search for a DTA that accepts the simplified tAPTA as an SMT formula. We evaluate our approach on randomly generated benchmarks and a scheduling case study. The results demonstrate the effectiveness of our simplification method in reducing the size of the encoded SMT formula and the efficiency of our approach in mining a DTA.
\end{abstract}


\input{introduction}

\input{preliminaries}

\input{problems}

\input{sEL}

\input{tAPTA}

\input{SMTmethod}

\input{guarantees}

\input{experiments}

\input{conclusion}

%\input{SMTmethod}

%\newpage
% \bibliographystyle{splncs04}
\bibliographystyle{IEEEtran}
\bibliography{reference}

% \newpage
 \input{appendix}

\end{document}

%% file: introduction.tex
\section{Introduction} \label{sc:introduction}

Real-time systems, ranging from autonomous driving controllers to industrial automation protocols, demand rigorous guarantees regarding their timing behaviors. A failure to meet strict timing constraints in such systems can result in catastrophic consequences. Formal verification has become an indispensable technique in the development of such safety-critical systems. Among various formal models, Timed Automata (TA) introduced by Alur and Dill~\cite{Alur94} are the de facto standard for modeling and verifying real-time behaviors~\cite{Franck2021Veri,Bouyer2005Model} due to their ability to capture both discrete logic and continuous time evolution across domains from scheduling algorithms \cite{Behrmann2005Schedule} and communication protocols \cite{Ouaknine2005MTL} to
safety-critical systems \cite{Jiang2021Safe}. Manually constructing TA models for small-scale systems is feasible, but it becomes increasingly error-prone and labor-intensive as system complexity grows, and even impossible for black-box AI-enabled safety-critical systems. This challenge has stimulated significant interest in automated model learning, a technique that automatically constructs formal models from observed system behaviors, typically in the form of timed execution traces. Model learning is useful for parameter design, analysis, and formal verification of real-time systems by estimating information of black-box or partially unknown real-time systems. For instance, over-approximated DTA mining results can be effectively utilized in assume-guarantee reasoning of timed systems \cite{ChenSZLM23}.

Both passive and active learning efforts have primarily concentrated on learning various subclasses of timed automata, e.g., event-recording automata~\cite{Grinchtein10,HenryJM20}, real-time automata ~\cite{VerwerWW12,SchmidtGHK13,CornanguerLRT22,AnWZZZ21,AnZZZ21,Meng20263DRTA}, timed automata with one-clock ~\cite{VerwerWW11,AnCZZZ20,XuAZ22}, Mealy machines with one timer~\cite{DierlHKKLLM23,VaandragerB021}, deterministic timed automata~\cite{TapplerALL19,AichernigPT20,TapplerAL22,Waga23,TengZA24}, etc. Active learning algorithms mainly extend the $L^*$ \cite{Angluin78} algorithm to the timed domain, e.g., \cite{Waga23,TengZA24} proposed active learning algorithms for DTA with provable convergence. 
%and complexity bounds are given. 
The correctness of active learning relies heavily on the capability of the teacher (oracle) module to answer equivalence queries (EQs), i.e., whether the hypothesized model is equivalent to the target system. 

However, in real-world engineering scenarios, particularly when dealing with black-box systems, legacy components, or safety-critical infrastructure, answering EQs is often impossible or prohibitively expensive. This reality necessitates a shift towards passive learning, which infers models solely from a fixed, finite dataset of observed system traces, without requiring further system interaction. Yet, passive learning faces its own challenges: lacking the corrective guidance of equivalence queries, it is highly sensitive to the quality and quantity of input data. Consequently, the learned model may be excessively over-approximated compared to the target model if the input data are insufficient. 

Existing passive methods for DTA, especially those leveraging SMT solvers \cite{TapplerAL22}, often struggle with scalability when processing massive raw datasets. When directly encoding large, redundant sets of raw system traces into an SMT formula, the resulting constraints can become computationally intractable, leading to solver timeouts or overflow errors. Although \cite{Meng20263DRTA} attempted to mitigate this by merging and simplifying input traces in DRTA learning, extending such merging strategies to general DTA remains difficult due to the complex delay times that must be considered across entire traces rather than single events. Therefore, the major challenges in passively learning timed automata include not only learning without a teacher, but also ensuring efficiency and enabling conflict detection when processing raw data, which motivates this work. Beyond SMT-based approaches, other techniques have been explored but suffer from their own limitations. Methods based on genetic programming \cite{TapplerALL19} or combining genetic algorithms with domain knowledge \cite{WallnerALT25} can handle larger search spaces but typically lack guarantees regarding termination, correctness, or minimality. Table~\ref{tab:rela} presents the comparison of the related work above.

Therefore, the primary motivation of this work is to bridge the gap between theoretical rigor and scalability. We aim to leverage the guarantees of the SMT-based framework while simultaneously ensuring efficiency and the capability to detect conflicts when handling massive raw datasets.
%our proposed approach.

\begin{table}[!t]
\caption{Comparison of Related Work.}
\begin{center}
\label{tab:rela}
% \vspace{-0.2cm}
\resizebox{\columnwidth}{!}{%
\setlength{\tabcolsep}{3pt}
\begin{tabular}{|c|c|c|c|c|}
\hline
 & Automata& learning type & method & limitation\\
\hline
\cite{Meng20263DRTA} & DRTA & Passive & SMT & Limited subclass \\
\hline
\cite{TapplerALL19} & DTA &  Passive & \makecell[c]{Genetic \\ programming } & \makecell[c]{No termination, correctness,\\  minimality guarantees} \\
\hline
\cite{TapplerAL22} & DTA & Passive&  SMT  &\makecell[c]{Excluding neg samples \\ \& limit on clocks} \\
\hline
\cite{Waga23} & DTA & Active & MN-congruence  & Limited minimality$^a$ \\
\hline
\cite{WallnerALT25} & DOTA & Passive & \makecell[c]{Genetic \\ programming \\ $\&$ knowledge }  & \makecell[c]{No termination, correctness,\\  minimality guarantees \\ \& limited subclass} \\
\hline
Ours & DTA & Passive  & SMT & Minimality$^b$ \\
\hline
\end{tabular}}
\end{center}
$^{\mathrm{a}}$ The learned DTA is the minimal DTA of the learned timed language congruence relation, not the minimal DTA recognizing the target timed language. The congruence relation typically induces a DTA that is much larger than the target DTA.

$^{\mathrm{b}}$ %The learned DTA is the minimal DTA recognizing the pair of timed languages of tAPTA, not the minimal DTA recognizing the given set of timed words.
When our simplification (Def.~\ref{def:ExMerge}, \ref{def:Intoverapp}) is applied to the tAPTA, the learned DTA is the minimal DTA recognizing the pair of timed languages of the tAPTA; whereas if the unsimplified tAPTA is used directly, it yields the minimal DTA recognizing the given set of timed words. 
% \vspace{-0.5cm}
\end{table}

\myparagraph{Contribution} The main contributions of this work are summarized as follows:
\begin{itemize}
    \item A trace simplification framework: We adapt simple elementary languages (sEL) \cite{Waga23} to preprocess traces, which removes redundant traces and detects potential conflicts. We then propose an incremental writing style for sEL to enable the construction of a timed augmented prefix tree acceptor (tAPTA) from sEL for further simplification and time-constraint over-approximation.
    \item SMT encoding scheme: We develop an encoding method that maps the simplified tAPTA into an SMT formulation for DTA synthesis without losing the expressiveness of guards, significantly reducing the computational complexity compared to direct encoding of raw data.

    \item Empirical validation: We implement our approach and evaluate it against existing SMT-based DTA synthesis approach on randomly generated benchmarks and a case study.
\end{itemize}

\myparagraph{Outline} In the following, Section~\ref{sc:preliminary} recalls DTA, APTA and the elementary language. An overview of the DTA mining problem and our mining approach is given in Section~\ref{sc:problems}. The approach is compared with existing work in Section~\ref{sc:compare} and empirically evaluated in Section~\ref{sc:experiments}. Finally, Section~\ref{sc:conclusion} concludes the paper.

% The omitted proofs of Lemma \ref{lemma:sEL}, \ref{lemma:ExMerge} are available in the full version~\cite{wang2026arXiv}.

%% file: preliminaries.tex
\section{Preliminaries} \label{sc:preliminary}
For the reader's convenience, the main notations used in this paper are listed in Table \ref{tab:notation}.

In this paper, we fix $\Sigma$ as a finite alphabet of letters (or events). We call an interval $\mathrm{I}\subset\mathbb{R}$ natural-bounded iff $\mathrm{I}= \langle a, b \rangle,\,  a\in \mathbb{N},\, b \in \mathbb{N}\cup\{+\infty\},\,\langle \in \{\ (, [\ \}, \rangle \in \{ \ ), ]\}$. Let $\mathbb{R}_{\geq 0}$ and $\mathbb{N}$ be the sets of non-negative reals and natural numbers. A \emph{timed word} is a finite sequence $\omega=(\sigma_1,t_1)(\sigma_2,t_2)\cdots(\sigma_n,t_n)$ $ \in (\Sigma\times\mathbb{R}_{\geq 0})^*$, where $t_i$ represents the delay time before the occurrence of event $\sigma_i\in\Sigma$ for all $1\leq i\leq n$. Particularly, the special symbol $\varepsilon$ represents the \emph{empty timed word}. The length of such timed word is $|\omega|=n$, with $|\varepsilon|=0$. For a timed word $\omega$, we define two functions $\lambda$ and $\mu$ to get the ordered sequence of delay times and events, respectively. For instance, let $\omega=(a,0.2)(b,0)(a,6)$, then we have $\lambda(\omega)=0.2,0,6$ and $\mu(\omega)=aba$. We denote by $\mathcal{T}\Sigma^*$ the set of timed words over $\Sigma$. A \emph{timed language} $L\subseteq\mathcal{T}\Sigma^*$ is a set of timed words. The concatenation of two timed words $u$, $v$ is denoted by $u\cdot v$ ($uv$, for short). The prefixes of a timed word $\omega$ is defined as $\mathbf{pre}(\omega)=\{u\in\mathcal{T}\Sigma^*| \exists v\in\mathcal{T}\Sigma^*, u\cdot v = \omega\}$, and extended to sets of timed words as $\mathbf{pre}(\Omega)=\bigcup_{\omega \in \Omega}\mathbf{pre}(\omega)$. 

Given timed languages $L, L'$, we say $L'$ is an over-approximation of $L$ iff $L\subseteq L'$. Given a pair of timed language $(L_+,L_-)$ with $L_+\cap L_- = \emptyset$, we say another pair of timed language $(L_+',L_-')$ is a safe over-approximation of $(L_+,L_-)$ if $L_+\subseteq L_+',\ L_-\subseteq L_-',\ L_+'\cap L_- = L_+\cap L_-' = \emptyset$. Given $(L_+,L_-)$ and a timed automaton $\mathcal{A}$ of which $L_+',L_-'$ are the accepting and rejecting timed languages of $\mathcal{A}$, we say $(L_+',L_-')$ is the language pair of $\mathcal{A}$, and $\mathcal{A}$ recognizes $(L_+,L_-)$ if $(L_+',L_-')$ is a safe over-approximation of $(L_+,L_-)$.

A timed automaton  (TA) is an extension of a finite state automaton with a set of real-valued clock variables measuring the time elapsed since transitions were taken. Each transition is constrained by some inequalities over the clock variables (i.e., guards), and once a transition is taken, it may reset some of the clocks to $0$. A guard clause $g$ over clock set $\mathcal{C}$ is of the form $ c\, \triangleright\, k $, with $c\in\mathcal{C} $, $\triangleright \in \{>,\geq,<,\leq \}$, $k\in\mathbb{N}$. A guard $g\in \mathcal{G}(\mathcal{C})$ is a finite conjunction of guard clauses. Formally, 

\begin{mydefinition}[Timed Automata\cite{Alur94}] \label{def:ta}
A TA is a tuple $\mathcal{A} =(Q,q_0, \mathcal{C},\Sigma,\delta, F)$, where $Q$ is the set of state, $q_0\in Q$ is the initial state, $\mathcal{C}$ is the set of clocks, $\Sigma$ is the finite alphabet of events, $F\subset Q$ is the set of accepting states, and $\delta\subseteq Q\times\Sigma\times\mathcal{G}(\mathcal{C})\times 2^{\mathcal{C}}\times Q$ is the set of transitions. A transition is a tuple $(q,a,g,r,q')\in \delta$, meaning that the transition is taken from $q$ to $q'$, with event $a$, guard $g$, and it resets the clocks $c\in r$. 

\end{mydefinition}

\begin{mydefinition}[Semantics of Timed Automata] \label{def:tasem}
For a TA $\mathcal{A}=(Q,q_0, \mathcal{C},\Sigma,\delta, F)$, the \emph{timed transition system} (TTS) is a $4$-tuple $(S,s_0,S_F,\tttrans)$, where
\begin{itemize}
\item $\nu:\mathcal{C}\rightarrow\mathbb{R}_{\geq 0}^{|\mathcal{C}|}$ is the clock valuation; $\nu \oplus \tau := \nu + \tau \cdot \mathbf{1}$ presents the clock valuation after a delay $\tau$; 
\item $S=Q\times(\mathbb{R}_{\geq 0})^{|\mathcal{C}|}$ is the set of combined states of the TA states and the clock values;
\item $s_0=(q_0,\mathbf{0^{|\mathcal{C}|}})$ is the initial state;
\item $S_F=\{(q,\nu)\in S\mid q\in F\}$ is the set of accepting states;
\item $\tttrans\;\subseteq S\times S$ is the combined transition of a delay $\tau$ and an event $a$. For each $(q,\nu),(q',\nu')\in S$, $\tau\in\mathbb{R}_{>0}$, and $(q,a,g,r,q')\in\delta$, we have 
\begin{align*}
(q,\nu){\xrightarrow{(a,\tau)}}(q',\nu')
\; \text{if}\;
& \nu\,\oplus\,\tau\models g,\;\text{and} \;\forall c\in r,\,\nu'(c)=0,\, \\ 
& \forall c \in \mathcal{C}\,\backslash \,r, \, \nu'(c)=\nu(c)+\tau
\end{align*}
\end{itemize}

A \emph{run} of a TA $\mathcal{A}$ is an alternating sequence $s_0,\tttrans_1,s_1,\tttrans_2,\cdots,\tttrans_n, s_n$ of states $s_i\in S$ and transitions $\tttrans_i\,\in\,\tttrans$ such that $s_{i-1}\tttrans_i s_i$ for every $i\in\{1,2,\dots,n\}$. A run is \emph{accepting} if $s_n\in S_F$. Given such a run, the associated \emph{timed word} is the concatenation of the labels of the transitions. A timed word $\omega$ associated with an accepting run is accepted by $\mathcal{A}$, denoted by $\mathcal{A}(\omega)=+$, and otherwise $\mathcal{A}(\omega)=-$. The \emph{timed language} $\mathcal{L}(\mathcal{A})$ of $\mathcal{A}$ is the set of timed words accepted by $\mathcal{A}$. %We say a TA $\mathcal{A}$ recognize a pair of timed language $(L_+,L_-)$ iff $L_+\subset\mathcal{L}(\mathcal{A})\wedge L_-\cap\mathcal{L}(\mathcal{A})=\emptyset$.
\end{mydefinition}

\begin{mydefinition}[Augmented Prefix Tree Acceptor (APTA) \cite{Francois1998APTA}] \label{def:apta}
A three-valued deterministic finite automaton (3-DFA) $\mathcal{P}$  is a transition system of the form 
$(\mathit{Q},q_0,\Sigma,\delta,F,R)$, 
%that marks states with $\{accepting,\, rejecting, \mathit{don't}$-$\mathit{care}\}$, 
where $Q$ is the set of states, 
which is partitioned into three disjoint subsets, i.e., 
$F$ the set of \emph{accepting} states, $R$ the set of \emph{rejecting} states, and $D=Q \backslash (F\cup R)$ the set of $\mathit{don't}$-$\mathit{care}$ states. 
$q_0$ is the initial state and $\Sigma$ is the finite alphabet of events. 
$\delta: Q \times \Sigma\rightarrow Q$ is the transition function. A transition over $\delta$ is a 3-tuple $(q,u,q')$ such that $\delta(q,u)=q'$. Given two sets of event sequences $S_+$, $S_-\subset\Sigma^*$, an augmented prefix tree acceptor (APTA) constructed from them is a 3-DFA 
%constructed from finite sets of positive and negative event sequences:  respected with them is 
$\mathcal{P}(S_+,S_-)=(\mathbf{pre}(S_+\cup S_-) ,\varepsilon,\Sigma,\delta_{\mathbf{pre}(S_+\cup S_-)},S_+,S_-)$, where for all $\sigma\in\Sigma^*$ and $u\in\Sigma$ such that $\sigma,\, \sigma u\, \in \mathbf{pre}(S_+\cup S_-)$, $\delta_{\mathbf{pre}(S_+\cup S_-)}(\sigma,u)=\sigma u$.
\end{mydefinition}

The states in an APTA are organized in a tree structure, with the transitions as edges. Throughout this paper, we utilize the terms $parent$, $child$, $ancestor$, and $descendant$ for the tree structure to represent the relation between states in an APTA. A path incoming to state $q_n$ in an APTA $\mathcal{P}=(\mathit{Q},\varepsilon,\Sigma,\delta,F,R)$ is an alternating sequence $p=q_0d_1q_1\cdots d_nq_n$ where $q_i\in Q$ for all $ 0\leq i\leq n$, and $d_i=(q_{i-1},u_i,q_i)$ is a transition over $\delta$. $p$ is a root path if $q_0=\varepsilon$. Throughout this paper, we focus exclusively on root paths and refer to them simply as "paths". $path(q)$ is the set of paths incoming to $q$. We also extend "prefix" to APTA paths: $\mathbf{pre}(p)=\{q_0d_1q_1\cdots d_iq_i\,| \, 1\leq i\leq n\}\cup\{q_0\}$. If $n>0$, the label sequence of $p$ is $L(p)=u_1u_2\cdots u_n$, otherwise $L(p)=\varepsilon$. 

An APTA's size grows drastically with the growth of the number of examples it is constructed from, since every prefix of an example will possibly be added as a unique state in the APTA. As an APTA is acyclic, as discussed in
\cite{10.1162/089120100561601}, a way to minimize an APTA is linear-time backward traversal and merging. To the end,  we need to identify the equivalent states during the backward traversal. 

\begin{mydefinition}[Equivalent States in APTA]\label{def:eqstate}
    Given an APTA $\mathcal{P}=(\mathit{Q},\varepsilon,\Sigma,\delta,F,R)$, an equivalence relation $\sim_S: \mathit{Q}\times\mathit{Q}$ is defined as:
    \begin{itemize}
        \item If two states $q_1$, $q_2\in Q$ are both in $F$ (or $R$, or $D=Q \backslash (F\cup R)$), and both have no children, then $q_1$ and $q_2$ are equivalent.
        \item  If two states $q_1$, $q_2\in Q$ are both in $F$ (or $R$, or $D$), and for all $\sigma \in \Sigma$, either both $\delta(q_1,\sigma)$ and $\delta(q_2,\sigma)$ exist and they are equivalent, or neither $\delta(q_1,\sigma)$ nor $\delta(q_2,\sigma)$ exists, then $q_1$ and $q_2$ are equivalent.
    \end{itemize}
    We write $q_1\equiv_S q_2$ iff they are equivalent under $\sim_S$.
\end{mydefinition}

Once a pair of equivalent states is identified, they can be merged into one representative state, and thus the APTA's size becomes smaller. When there are no equivalent states left to merge, we say the APTA is simplified. A simplified APTA is a minimal 3-DFA that recognizes the example sequences the original APTA is constructed from.

\begin{mydefinition}[Elementary Language~\cite{Waga23}] \label{def:el}
For a timed word $\omega=(\sigma_1,t_1)(\sigma_2,t_2)\cdots(\sigma_n,t_n)$, let $T_{j,k}(\omega)=\sum\nolimits_{i=j}^{k}{t_i}, j,k=1,\cdots,n$. A time condition is a finite conjunction of constraints $\Lambda=\bigwedge (T_{j,k}(\omega)\in \mathrm{I}_{j,k})$  over $ \lambda(\omega)$, where $\mathrm{I}_{j,k}$ are natural-bounded intervals. A time language $L$ is $elementary$ if there exist $u\in \Sigma^*$ and a time condition $\Lambda$, such that $L=\{\omega\in\mathcal{T}\Sigma^*\,|\,\mu(\omega)=u\, \wedge\, \lambda(\omega)\models\Lambda\}$. We extend the function $\mu$ to the elementary language such that $\mu(L)=u$. We denote such an elementary language $L$ by $\mathcal{E}(u, \Lambda)$, and the set of elementary language over $\Sigma$ by $\mathcal{E}(\Sigma)$. 

\end{mydefinition}

\begin{table}[h]\label{tab:notation}
\caption{Notation Table}
\begin{tabularx}{\linewidth}{p{0.12\linewidth}X}
\toprule
\textbf{Notation} & \textbf{Meaning} \\
\midrule
$\Sigma$ & Alphabet of events\\
$\lambda(\omega)$ & Delay time sequence of $\omega$ \\
$\mu(\omega)$ & Event sequence of $\omega$\\
$\mathcal{T}\Sigma^*$  & The set of timed words over $\Sigma$ \\
$\mathbf{pre}()$ & the set of prefixes of a timed word (or a set of timed words) \\
$\mathcal{C}$ & Clock set\\
$\mathcal{G}(\mathcal{C})$ & The set of guards over $\mathcal{C}$\\
$\mathcal{A}$ & Time automaton\\
$\mathcal{P}$ & APTA or tAPTA\\
$\delta$ & Transitions\\
$F$ & Accepting states\\
$R$ & Rejecting states\\
$D$ & $\mathit{Don't}$-$\mathit{care}$ states\\
$path(q)$ & The set of paths incoming to APTA state $q$\\
$\Lambda$ & Time condition\\
$\mathcal{E}(u, \Lambda)$ & The elementary language over $(u,\Lambda)$\\
$\mathcal{E}(\Sigma)$ & The set of elementary languages over $\Sigma$\\
$\mathcal{SE}(\omega)$ & The simple elementary language corresponding to $\omega$\\
$\mathcal{SE}(\Sigma)$ & The set of simple elementary languages over $\Sigma$\\
$\Omega$ & Example (trace) set\\
$S(\Omega)$ & The set of sELs translated from $\Omega$\\
$P(S)$ & The simplified tAPTA constructed from $S$\\

\bottomrule
\end{tabularx}
\end{table}

%% file: problems.tex
\section{Overview of DTA Mining} \label{sc:problems}
In this section, we define the DTA mining problem and the minimal DTA mining problem, and then we describe our approach to solving it.%elucidate 
%discuss several distinctions between synthesizing TRE and traditional RE.

\begin{prob}[DTA mining problem]\label{prob:DTAmining}
Given a set of example timed words $\Omega=(\Omega_{+},\Omega_{-})$, find a DTA $\mathcal{A}$ such that for all $\omega \in \Omega_{+}, \mathcal{A}(\omega)=+$, and for all $\omega \in \Omega_{-}, \mathcal{A}(\omega)=-$, or claim that there does not exist such a DTA. The set of such DTAs is denoted as $\mathit{A}(\Omega)=\mathit{A}(\Omega_{+},\Omega_{-})$.
\end{prob}

Normally we want to find a minimal solution DTA, i.e., a DTA with the least number of states among the solutions to Problem \ref{prob:DTAmining}:

\begin{prob}[Min-DTA mining problem]\label{prob:mDTAmining}
Given a set of example timed words $\Omega=(\Omega_{+},\Omega_{-})$, find a minimal DTA (Min-DTA) $\mathcal{D}$ such that for all $\omega \in \Omega_{+}, \mathcal{D}(\omega)=+$, and for all $\omega \in \Omega_{-}, \mathcal{D}(\omega)=-$, or claim that there does not exist such DTA. The set of such minimal DTAs is denoted as $\mathit{D}(\Omega)=\mathit{D}(\Omega_{+},\Omega_{-})$.
\end{prob}

Where no confusion arises, we extend the notation $\mathit{A}(\cdot,\cdot)$ and $\mathit{D}(\cdot,\cdot)$ to any pair of positive and negative timed languages in what follows, to present the DTAs and Min-DTAs mined from the pair of timed languages. By definition, we have:

\begin{mylemma}\label{lemma:minDTA}
If $\mathit{A}(\Omega_{+},\Omega_{-}) = \mathit{A}(\Omega_{+}',\Omega_{-}')$ then $\mathit{D}(\Omega_{+},\Omega_{-}) = \mathit{D}(\Omega_{+}',\Omega_{-}')$. 
% \hfill $\square$
\end{mylemma}

Our approach mainly consists of the following 3 steps:
\begin{itemize}
    \item Preprocessing the examples: we translate every example timed word (trace) to simple elementary timed languages (sEL)\cite{Waga23}. Then we discard the redundant examples and check for conflicts. If there are any conflicts, return $\mathit{A}(\Omega_{+},\Omega_{-})=\emptyset$.
    \item Simplifying the examples with tAPTA: we transform sELs into incremental-formed sequences and construct an APTA from them (called tAPTA), and then we simplify the tAPTA with extended merging analogous to APTA simplification.
    \item Encoding a minimal DTA SMT problem of given DTA size setting. If it is UNSAT, try the next larger setting. If it is SAT, construct a DTA ${\mathcal{A}}$ according to the model and return ${\mathcal{A}}$.

\end{itemize}

%% file: sEL.tex
\subsection{Preprocessing Examples} \label{sbsc:preproc}
%Before the repair problem is defined, we should recall and extend the notion of distance on PBE-based repair of RE. The distance is to quantify the cost of edit to rewrite the reference to a candidate. In RE cases, only the replacement of nodes in subtrees is counted. However, in TRE cases the edit involves both replacement of nodes and the changes on time interval. We posit that it is more desirable to alter time intervals while maintaining the nodes unchanged to achieve the repair. Stemming from this insight, we provide the following definition:
In this subsection, we transform the given examples to sEL, such that redundancy (removal of some examples from the set yields an equivalent DTA regardless of the mining approach) can be removed, and conflicts that make the problems unsolvable ($\mathit{A}(\Omega_{+},\Omega_{-})=
%\mathit{D}(\Omega_{+},\Omega_{-})=
\emptyset$) will be detected if any exist. 

\begin{mydefinition}[Simple Elementary Language~\cite{Waga23}] \label{def:sel}
Given a timed word $\omega_0=(\sigma_1,t_1)(\sigma_2,t_2)\cdots(\sigma_n,t_n)$,  the corresponding \emph{simple elementary language} (sEL) is an elementary language $\mathcal{SE}({\omega_0}) =  \{\omega=(\sigma_1,t_1')(\sigma_2,t_2')\cdots(\sigma_n,t_n')|t_1',t_2',\\\cdots,t_n' \models \Lambda_{\omega_0}\}$, where $\Lambda_{\omega_0}=\underset{{\theta\in\Theta_{\omega_0}}}{\bigwedge}\theta$.  $\Theta_{\omega_0}$ is the finite set of the tightest integer-bounded time constraints such that $\lambda(\omega_0) \models \Lambda_{\omega_0}$, i.e., $\forall 1\le j\le k\le n$, either $(d < T_{j,k}(\omega)< d+1) \in \Theta_{\omega_0}$ or $(T_{j,k}(\omega) = d) \in \Theta_{\omega_0}$ with $d\in\mathbb{N}$. The set of all sELs over $\Sigma$ is denoted by $\mathcal{SE}(\Sigma)$.
\end{mydefinition}

\begin{myexample}
The sEL of $\omega=(a,1.6)(b,2.6)(a,0.4)$ is $\mathcal{SE}({\omega}) =  \{(a,t_1)(b,t_2)(a,t_3)|t_1\in(1,2)\wedge t_2\in(2,3)\wedge t_3\in(0,1)\wedge t_1+t_2\in(4,5)\wedge t_2+t_3=3\wedge t_1+t_2+t_3\in(4,5)\}$.
\end{myexample}

\begin{mylemma}\label{lemma:sEL}
Given timed words $\omega_1$, $\omega_2$, $\mathcal{SE}({\omega}_1)= \mathcal{SE}({\omega}_2)$ iff for any timed automaton $\mathcal{A}$, $\omega_1\in L(\mathcal{A}) \iff \omega_2\in L(\mathcal{A})$.
\end{mylemma}

The proof of this lemma is provided in Appendix \ref{app:proof}.

By Lemma \ref{lemma:sEL}, after translating all given examples to their corresponding sEL, redundancy can be omitted and conflicts can be detected. 

\begin{proposition}\label{prop:sEL}
Given a set of timed words $\Omega=(\Omega_{+},\Omega_{-})$ and $\omega_1\in\Omega_+$, $\omega_1'\in\Omega_-$:
\begin{itemize}
    \item Redundancy: For all $\omega_2\neq \omega_1,\omega_1'$, $\mathit{A}(\Omega_{+},\Omega_{-})=\mathit{A}(\Omega_{+}\backslash \{\omega_1\},\Omega_{-})$ if $\omega_2\in \Omega_+ \wedge \mathcal{SE}({\omega}_1)= \mathcal{SE}({\omega}_2)$; $\mathit{A}(\Omega_{+},\Omega_{-})=\mathit{A}(\Omega_{+},\Omega_{-}\backslash \{{\omega_1'}\})$ if $\omega_2\in \Omega_- \wedge \mathcal{SE}({\omega}_1')= \mathcal{SE}({\omega}_2)$. 
    \item Conflict: $\mathit{A}(\Omega_{+},\Omega_{-})=\emptyset$ if $\mathcal{SE}({\omega}_1)= \mathcal{SE}({\omega}_1')$.
\end{itemize}
\end{proposition}

\begin{proof}
Given $\Omega=(\Omega_{+},\Omega_{-})$ and $\omega_1\in\Omega_+$, $\omega_1'\in\Omega_-$:
\begin{itemize}
    \item Redundancy: For the first case, if $\omega_2\in \Omega_+ \wedge \mathcal{SE}({\omega}_1)= \mathcal{SE}({\omega}_2)$, then for all $\mathcal{A}\in\mathit{A}(\Omega_{+},\Omega_{-})$ and $\mathcal{A}'\in\mathit{A}(\Omega_{+}\backslash \{\omega_1\},\Omega_{-})$, obviously $\mathcal{A}\in A(\Omega_{+}\backslash \{\omega_1\},\Omega_{-})$ and by Lemma \ref{lemma:sEL}$, \mathcal{A}'\in A(\Omega_{+},\Omega_{-})$. Hence $\mathit{A}(\Omega_{+},\Omega_{-})=\mathit{A}(\Omega_{+}\backslash \{\omega_1\},\Omega_{-})$. The proof of the second case is similar.
    \item Conflict: $\mathcal{SE}({\omega}_1)= \mathcal{SE}({\omega}_1')$. If $\mathit{A}(\Omega_{+},\Omega_{-})\neq\emptyset$, then for all $\mathcal{A}\in\mathit{A}(\Omega_{+},\Omega_{-})$, either $\mathcal{A}(\omega_1)=\mathcal{A}(\omega_1')=-$ or $\mathcal{A}(\omega_1')=\mathcal{A}(\omega_1)=+$. Each case yields a contradiction. Hence $\mathit{A}(\Omega_{+},\Omega_{-})=\emptyset$.
\end{itemize}
\end{proof}

%\begin{corollary}\label{cor:conflictfree}
%$\Omega=(\Omega_{+},\Omega_{-})$ is conflict-free if and only if $\forall \omega_1\in\Omega_+$, $\omega_1'\in\Omega_-$, $\mathcal{SE}({\omega}_1)\neq \mathcal{SE}({\omega}_1')$. And a conflict-free $\Omega$ is guaranteed to have a solution to Problem \ref{prob:DTAmining}.
%\end{corollary}

%The proof of this corollary is provided in Appendix \ref{app:proof}.

Given $\Omega=(\Omega_{+},\Omega_{-})$, if there is any redundancy in $\Omega$, we just discard the duplicate sELs. If there is any conflict detected (as characterized in Prop. \ref{prop:sEL}) in $\Omega$, we claim there is no such DTA for Problem \ref{prob:DTAmining}. For simplicity, we assume that there are no conflicts in $\Omega$ in what follows. $S(\Omega)=(S_+(\Omega),S_-(\Omega))$  where $S_+(\Omega)=\{\mathcal{SE}(\omega)|\omega\in\Omega_+\}$, $S_-(\Omega)=\{\mathcal{SE}(\omega)|\omega\in\Omega_-\}$ is the set of sELs translated from $\Omega=(\Omega_{+},\Omega_{-})$.

\begin{mylemma} \label{lemma:sEL2}  
$S(\Omega)$ is a safe over-approximation of $\Omega$. 
% \hfill $\square$
\end{mylemma}
The proof is straightforward from the definition of sEL.

\begin{mytheorem} \label{thm:sEL}
$\mathit{A}(S(\Omega))=\mathit{A}(\Omega)$ and $\mathit{D}(S(\Omega))=\mathit{D}(\Omega)$.
\end{mytheorem}

\begin{proof}
For all $\mathcal{A}\in\mathit{A}(S(\Omega))$, it is obvious that $\mathcal{A}\in\mathit{A}(\Omega)$. 

For all $\mathcal{A}'\in\mathit{A}(\Omega)$, if $\mathcal{A}'\notin\mathit{A}(S(\Omega))$, then either there exists $\omega\in S_+(\Omega), s.t. \mathcal{A}'(\omega)=-$, or there exists $\omega\in S_-(\Omega), s.t. \mathcal{A}'(\omega)=+$. In the first case, by definition there exist $\omega'\in \Omega+, $, $ s.t. \, \mathcal{SE}({\omega})= \mathcal{SE}({\omega}')$. $\mathcal{A}'(\omega')=+ $ contradicts Lemma \ref{lemma:sEL}. In the second case,  by definition there exist $\omega'\in \Omega-, $, $ s.t. \, \mathcal{SE}({\omega})= \mathcal{SE}({\omega}')$. $\mathcal{A}'(\omega')=- $ contradicts Lemma \ref{lemma:sEL}. So $\mathcal{A}'\in\mathit{A}(S(\Omega))$.

Hence, $\mathit{A}(S(\Omega))=\mathit{A}(\Omega)$, and by Lemma \ref{lemma:minDTA}, we have $\mathit{D}(S(\Omega))=\mathit{D}(\Omega)$. 
\end{proof}

\begin{myremark}\label{rmk:sEL}
In general, there are very few redundant traces, as the number of possible sELs for a timed word of length $n$, alphabet $\Sigma$ and the minimal width $w$ of the range of each delay time exceeds $w^n|\Sigma|^n$ . We introduce sEL mainly because it %is the minimal MN-Equivalence class in the timed language of TA as claimed by Lemma \ref{lemma:sEL}, and it 
replaces the continuous delay times with integer-bounded intervals, which is suitable for further simplification.
\end{myremark}

%% file: tAPTA.tex
\subsection{Construction and Simplification of tAPTA }\label{sbsc:tAPTA}

\subsubsection{Constructing a timed APTA from sELs}
In the context of TA, a clock variable records the time elapsed since its most recent reset (or since the initial moment if never reset). Its value is the sum of delay times accumulated since that reset. The feasibility of a transition depends not only on the input event, but also on the value of these accumulated delay times. Therefore, we propose an incremental representation of elementary timed languages to explicitly record the newly introduced timing constraints of the accumulated delay times at each event.
%There is an observation of TA: the value of any clock variable records the amount of time that has elapsed since its most recent reset along some transition (or since the initial moment if the clock has never been reset), i.e., the time duration accumulated over the sequence of transitions taken since that reset. Whether a transition can be taken or not depends only on the event and the accumulated durations. Thus we propose an incremental form of elementary timed language that record newly introduced time constraints at every event:

\begin{mydefinition}[Incremental Form of Elementary Language] \label{def:IFTL}
Given an elementary language $\mathcal{E}(u,\Lambda)$ with $u=\sigma_1 \sigma_2 \cdots\sigma_n$ and $\Lambda=\{\sum\nolimits_{i=j}^{k}{t_i}\in \mathrm{I}_{j,k} | j,k=1,\cdots,n, \mathrm{I}_{j,k}$ is natural-bounded interval $ \}$, the incremental form of $\mathcal{E}(u,\Lambda)$ is a sequence of $n$ tuples (called incremental tuple). The $m$-th tuple is $(\sigma_m,\Lambda_m)$, where $\Lambda_m=\{T_{j,m}(\cdot)\in \mathrm{I}_{j,m} | j=1,\cdots,m \} \subseteq \Lambda $. That is, the $m$-th tuple contains the $m$-th event and the set of constraints of the sum of time delays at the $m$-th event. 
    
\end{mydefinition}

We impose a fixed canonical ordering on the elements of $\Lambda_m$ by increasing number of summands (i.e., $T_{m,m}\in \mathrm{I}_{m,m}$ comes the first and $T_{1,m}\in \mathrm{I}_{1,m}$ is the last).  Then we omit the $\in$ notation and abbreviate $\Lambda_m$ as the ordered list $(\mathrm{I}_{m,m},\mathrm{I}_{m-1,m},\cdots,\mathrm{I}_{1,m})$ in what follows. Furthermore, for sELs, since each interval $\mathrm{I}_{j,m}$ is either $[d,d]$ or $(d, d+1)$, $d\in \mathbb{N}$, we abbreviate these two types of elements as $d$ and $d_+$ respectively.

\begin{myexample}
Given $\omega=(a,1.6)(b,2.6)(a,0.4)$ and $\mathcal{SE}({\omega}) =  \{(a,t_1)(b,t_2)(a,t_3)|t_1\in(1,2)\wedge t_2\in(2,3)\wedge t_3\in(0,1)\wedge t_1+t_2\in(4,5)\wedge t_2+t_3=3\wedge t_1+t_2+t_3\in(4,5)\}$. The incremental form of $\mathcal{SE}({\omega})$ is $(a,(1_+))(b,(2_+,4_+))(a,(0_+,3,4_+))$.
\end{myexample}

The purpose of the incremental form is to express an elemental language with $\mathcal{O}(n^2)$ time constraints in a sequential manner, where each step incorporates all time constraints newly introduced by the current event. We treat each incremental formed sEL as an label sequence, where each tuple $(\sigma_m,\Lambda_m)$ serves as a single label. Consequently, we extend the notion of prefix to these label sequences. To avoid ambiguity, we refer to them as "sEL sequences", from which we construct an APTA denoted as the "timed APTA" (tAPTA).

\begin{mydefinition}[timed APTA (tAPTA)] \label{def:tAPTA}
    Given a pair of sets of sELs $S=(S_+,S_-)$ in incremental form, a tAPTA constructed from $S$ is an APTA $\mathcal{P}(S)=(Q,\varepsilon,\Sigma,\delta,S_+,S_-)$, where $Q=\mathbf{pre}(S_+\cup S_-)$ and $\delta: Q \times (\Sigma \times \mathbf\Lambda )\rightarrow Q$. $(\Sigma\times \mathbf\Lambda)$ is the set of incremental tuples. For all $\sigma\in\Sigma^*$ , $u\in\Sigma$ and $ \Lambda_1,\Lambda_2 \in \mathbf\Lambda$ such that $\mathcal{E}(\sigma,\Lambda_1),\, \mathcal{E}(\sigma u,\Lambda_1\cup\Lambda_2)\, \in Q$,   $\delta(\mathcal{E}(\sigma,\Lambda_1),(u,\Lambda_2))=\mathcal{E}(\sigma u,\Lambda_1\cup\Lambda_2)$.
    
\end{mydefinition}

The semantics of tAPTA are defined as follows. For a path $p$ in the tAPTA, its timed language  $L(p)$ corresponds to the elementary language represented by the incremental-formed sequence of labels along the transitions of that path. The timed language of a state $q$ is $L(q)=\underset{p\in path(q)}{\bigcup} L(p)$. A tAPTA is a three-valued transition system, and the language pair of a tAPTA $\mathcal{P}=(Q,\varepsilon,\Sigma,\delta,F,R)$ is $P=(\underset{q\in F}{\bigcup} L(q),\underset{q\in R}{\bigcup} L(q) )$. 

\oomit{
\begin{mylemma}
    Given a path $p$ in a tAPTA and timed word $\omega\in L(p)$, for all $p'\in\mathbf{pre}(p)$, there exists $\omega'\in\mathbf{pre}(\omega) $, such that $\omega'\in L(p')$.
\end{mylemma}

\begin{proof}
    Suppose $p=q_0d_1q_1\cdots d_nq_n$ and $\omega=(\sigma_1,t_1)(\sigma_2,t_2)\cdots(\sigma_n,t_n)$, where $d_i=(q_{i-1},u_i,q_i)$ for $1\leq i\leq n$, $u_i=(a_i,\Lambda_i)$ is an incremental tuple $(\sigma_i,\Lambda_i)$. By definition, $\omega\in L(p)\Rightarrow \sigma_1\cdots\sigma_n=a_1\cdots a_n,\wedge \lambda(\omega)\models \underset{1\leq i\leq n}{\bigwedge}\Lambda_i$. For all $p'\in\mathbf{pre}(p)$, if $p'=q_0$ then there exists $\omega'=\varepsilon \in \mathbf{pre}(\omega)$, $\omega'\in L(p')$. Otherwise let $p'=q_0d_1q_1\cdots d_mq_m,\, m\leq n$, $L(p')=\{v\in\mathcal{T}\Sigma^*\,|\mu(v)=a_1a_2\cdots a_m\wedge \lambda(v)\models \underset{1\leq i\leq m}{\bigwedge}\Lambda_i\}$, then there exists $\omega'=(\sigma_1,t_1)(\sigma_2,t_2)\cdots(\sigma_m,t_m)$. By definition \ref{def:IFTL}$, \underset{1\leq i\leq m}{\bigwedge}\Lambda_i$ contains only the time constraints over $m$ terms of delay time. We have: $$\lambda(\omega)\models \underset{1\leq i\leq n}{\bigwedge}\Lambda_i\Rightarrow t_1\cdots t_n\models \underset{1\leq i\leq m}{\bigwedge}\Lambda_i \Rightarrow t_1\cdots t_m\models \underset{1\leq i\leq m}{\bigwedge}\Lambda_i \Rightarrow \lambda(\omega')\models \underset{1\leq i\leq m}{\bigwedge}\Lambda_i$$
    Thus $\omega'\in L(p')$. 
\end{proof}
}

\subsubsection{Extending States Merging to tAPTA}

The construction of a tAPTA from sELs preserves comprehensive information about event sequences and time constraints. However, applying standard APTA state merging is challenging in this context, as transitions are labeled by tuples of events and time constraints rather than simply events.  We consequently extend the merging operation of APTA to tAPTA: % sEL sequence and show a merged state of the extended merging operation is still a timed language:

\begin{mydefinition}[Extended Merging of tAPTA] \label{def:ExMerge}
    Given a tAPTA $\mathcal{P}=(\mathit{Q},\varepsilon,\Sigma,\delta,F,R)$, an equivalence relation $\sim_S: \mathit{Q}\times\mathit{Q}$ is defined as:
    \begin{itemize}
        \item If two states $q_1$, $q_2\in Q$ are both in $F$ (or $R$, or $D=Q \backslash (F\cup R)$), and they are both leaf states (have no children), then $q_1$, $q_2$ are equivalent.
        \item If two states $q_1$, $q_2\in Q$ are both in $F$ (or $R$, or $D$), and for all $\sigma \in \Sigma$, 
        \begin{enumerate}
            \item $\nexists(\sigma,\Lambda_1)$, $s.t.\, \delta(q_1,(\sigma,\Lambda_1))$ exist $\Longleftrightarrow$ $\nexists(\sigma,\Lambda_2)$, $s.t.\,\delta(q_2,(\sigma,\Lambda_2))$ exist.
            \item $\exists(\sigma,\Lambda_1)$, $s.t.\,\delta(q_1,(\sigma,\Lambda_1))=q_1'$ $\Longrightarrow$ $\exists(\sigma,\Lambda_2)$, $s.t.\,\delta(q_2,(\sigma,\Lambda_2))=q_2'$ and $q_1',q_2'$ are equivalent.
            \item $\exists(\sigma,\Lambda_1)$, $s.t.\,\delta(q_2,(\sigma,\Lambda_2))=q_2'$ $\Longrightarrow$ $\exists(\sigma,\Lambda_1)$, $s.t.\,\delta(q_1,(\sigma,\Lambda_1))=q_1'$ and $q_1',q_2'$ are equivalent.
        \end{enumerate}
    \end{itemize}
    We write $q_1\equiv_S q_2$ iff they are equivalent under $\sim_S$.

    The extended merging operation of two equivalent states $q_1,q_2$ on a tAPTA is defined as follows:

    \begin{itemize}
        \item If  $q_1,q_2$ both have no children, then they are merged into one state.
        \item If  $q_1,q_2$ have any pair of equivalent children, i. e., $\exists \delta(q_1,(\sigma,\Lambda_1))=q_1'$, $\delta(q_2,(\sigma,\Lambda_2))=q_2'$, $q_1'\equiv_S q_2'$, we first operate the merging of $q_1',q_2'$ into $q'$, then we merge $q_1,q_2$ into $q$. The transitions from $q$ to $q'$ is labeled by tuple $(\sigma, \Lambda_1\vee \Lambda_2)$. 
    \end{itemize}
    
\end{mydefinition}

 The label of a merged transition takes the form $ (\sigma,(I_1,I_2,\\\cdots,I_n)\vee (I_1',I_2',\cdots,I_m'))$, meaning that for any path $p$ containing this transition, either $ (I_1,I_2,\cdots,I_n)$ or $ (I_1',I_2',\cdots,I_m')$ is selected based on the position of the transition in the path. If $n\neq m$, then only the branch of whose length matches the sequential number of this transition in $L(p)$ can be taken. For instance, in Fig. \ref{fig:ExMerge}.c, for the path $q_0 \rightarrow q_2 \rightarrow q_4$, the transition from $q_2$ to $q_4$ is the second transition and corresponds to the branch $(3_+,6_+)$, and the time language of the entire path is $((b,(3)),(a,(3_+,6_+)))$, identical to that of the path $q_0 \rightarrow q_2 \rightarrow q_5$ in Fig. \ref{fig:ExMerge}.a. Semantically, such a merged transition is equivalent to two distinct transitions sharing the same source and target states, labeled by $(\sigma,(I_1,I_2,\cdots,I_n))$ and $(\sigma,(I_1',I_2',\cdots,I_m'))$ respectively.

\begin{mylemma} \label{lemma:ExMerge}
The extended merging is "safe", i.e., no conflicts are introduced between the timed languages of accepted and rejected states.
\end{mylemma}

The proof of this lemma is provided in Appendix \ref{app:proof}.

\begin{myremark}\label{rmk:ExMerge}
Without extending the merging operation for tAPTA and relying solely on the merging defined in Def. \ref{def:eqstate}, only transitions labeled with identical tuples of events and time constraint lists can be merged. While this strict merging ensures that no over-approximation is introduced to the timed language, it limits model simplification. This framework differs from the tAPTA approach for RTA mining in \cite{Meng20263DRTA}, which retains time constraint information for only a single delay time. 
\end{myremark}

\subsubsection{Interval Over-approximation for Merged Transitions}

While extended merging effectively reduces the number of states and transitions, it preserves all time constraint information within the disjunctive branches of the merged transitions. To mitigate the resulting complexity of the SMT formulas used for encoding these time constraints, we apply an interval over-approximation to each merged transition immediately after merging a group of equivalent states:

\begin{mydefinition}[Interval Over-approximation] \label{def:Intoverapp}
For a merged transition of tuple $(\sigma, (I_1,I_2,\cdots,I_n)\vee (I_1',I_2',\cdots,I_m'))$,  for $1\leq i\leq min\{n,m\}$, we check whether $I_i$, $I_i'$ and the interval between $I_i$ and $ I_i'$ (denoted $I_i''$) overlap with other intervals that are from another transition with the same parent and event, and at the same index of the interval list:
\begin{itemize}
    \item If there exists $i$ such that $I_i$ or $ I_i'$ overlaps, we do nothing to the two branches.
    \item If for all $i$, neither $I_i$ nor $ I_i'$ overlaps, we do over-approximation $(I_1,I_2,\cdots,I_n)\vee (I_1',I_2',\cdots,I_m')\hookrightarrow(I_1\cup I_1',I_2\cup I_2',\cdots,I_m\cup I_m',I_{m+1},\cdots, I_n)$ (suppose $n>m$).
    \item Based on the second scenario, for all $i$,  if $I_i''$ is not overlapped, we further do over-approximation $I_i'\cup I_i'\hookrightarrow I_i'\cup I_i'\cup I_i''$, i.e. combine $I_i$, $ I_i'$  and the interval between them into one. 
\end{itemize}
\end{mydefinition}

Fig.~\ref{fig:ExMerge}.d illustrates this over-approximation. When composing the timed language of a path containing an over-approximated transition with a sequential number $m$, only the first $m$ intervals in the over-approximated list is taken, as the first $m$ intervals form an over-approximation of a disjunction branch of length $m$.  For instance, the timed language of the path via $q_0 \rightarrow q_2 \rightarrow q_3$ in Fig \ref{fig:ExMerge}.d is $(b,(3))(a,((1,3),(2,6)))$.

\begin{myexample}\label{exp:ExMerge}
    Given $S_{+}=\{(a,(1))(b,(1,2))(a,(1_+,2_+,3_+))\\,(b,(3))(a,(2_+,5_+)\}$, $S_{-}=\{(a,(1))(b,(1,2))(a,(0_+,1_+,\\2_+)),(b,(3))(a,(3_+,6_+))\}$, first we merge the four leaf states. Then we merge the states for prefixes $(a,(1))(b,(1,2))$ and $(b,(3))$. Fig. \ref{fig:ExMerge} illustrates the simplification of the tAPTA.  
\end{myexample}

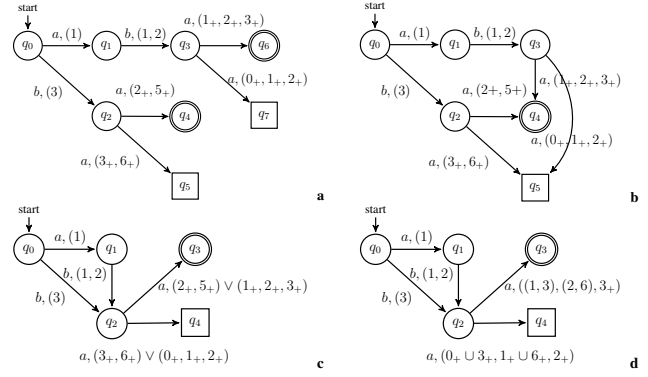
\begin{figure}[!t]
    \begin{minipage}{0.48\linewidth}
    \centering
    \resizebox{\linewidth}{!}{
    \begin{tikzpicture}[->, >=stealth', shorten >=1pt, auto, node distance=1.5cm, thick,scale=0.8, every node/.style={scale=0.8}, initial where=above]
    \node[initial, state] (0) {\large $q_0$};
    \node[state] (1) [right = 1.2cm of 0] {\large $q_1$};
    \node[state] (2) [below = 1cm of 1] {\large $q_2$};
    \node[state] (3) [right = 1.2cm of 1] {\large $q_3$};
    \node[state,accepting] (4) [below = 1cm of 3] {\large $q_4$};
    \node[rectangle,draw,minimum width =0.8cm,minimum height = 0.8cm] (5) [below = 1cm of 4] {\large $q_5$};
    \node[state,accepting] (6) [right = 1.2cm of 3] {\large $q_6$};
    \node[rectangle,draw,minimum width =0.8cm,minimum height = 0.8cm] (7) [below = 1cm of 6] {\large $q_7$};
    \path (0) edge node[above] {\large $a,(1)$} (1)
    (1) edge   node[above] {\large $b,(1,2)$} (3)
    (0) edge   node[below left ] {\large $b,(3)$} (2)
    (2) edge   node[below left ] {\large $a,(3_+,6_+)$} (5)
    (3) edge   node[right ] {\large $a,(0_+,1_+,2_+)$} (7)
    (2) edge   node[above=0.35cm] {\large $a,(2_+,5_+)$} (4)
    (3) edge  node[above=0.35cm] {\large $a,(1_+,2_+,3_+)$} (6);
    \end{tikzpicture}
    \textbf{a}
    }
    \end{minipage}
    \begin{minipage}{0.48\linewidth}
    \centering
    \resizebox{0.9\linewidth}{!}{
    \begin{tikzpicture}[->, >=stealth', shorten >=1pt, auto, node distance=1.5cm, thick,scale=0.8, every node/.style={scale=0.8}, initial where=above]    \node[initial, state] (0) {\large $q_0$};
    \node[state] (1) [right = 1.2cm of 0] {\large $q_1$};
    \node[state] (2) [below = 1cm of 1] {\large $q_2$};
    \node[state] (3) [right = 1.2cm of 1] {\large $q_3$};
    \node[state,accepting] (4) [below = 1cm of 3] {\large $q_4$};
    \node[rectangle,draw,minimum width =0.8cm,minimum height = 0.8cm] (5) [below = 1cm of 4] {\large $q_5$};
    \path (0) edge node[above] {\large $a,(1)$} (1)
    (1) edge   node[above] {\large $b,(1,2)$} (3)
    (0) edge   node[below left ] {\large $b,(3)$} (2)
    (2) edge   node[below left ] {\large $a,(3_+,6_+)$} (5)
    (3) edge   node[right ] {\large $a,(1_+,2_+,3_+)$} (4)
    (2) edge   node[above=0.35cm] {\large $a,(2+,5+)$} (4)
    (3) edge [in= 45, out=-45] node[below=0.35cm] {\large $a,(0_+,1_+,2_+)$} (5);
    \end{tikzpicture}
    \textbf{b}
    }
    \end{minipage}
    \newline
    \begin{minipage}{0.48\linewidth}
    \centering
    \resizebox{\linewidth}{!}{
    \begin{tikzpicture}[->, >=stealth', shorten >=1pt, auto, node distance=1.5cm, thick,scale=0.8, every node/.style={scale=0.8}, initial where=above]\node[initial, state] (0) {\large $q_0$};
    \node[state] (1) [right = 1.2cm of 0] {\large $q_1$};
    \node[state] (2) [below = 1cm of 1] {\large $q_2$};
    \node[state,accepting] (3) [right = 1.2cm of 1] {\large $q_3$};
    \node[rectangle,draw,minimum width =0.8cm,minimum height = 0.8cm] (4) [below = 1cm of 3] {\large $q_4$};
    \path (0) edge node[above] {\large $a,(1)$} (1)
    (0) edge   node[below left ] {\large $b,(3)$} (2)
    (2) edge   node[right ] {\large $a,(2_+,5_+)\vee(1_+,2_+,3_+)$} (3)
    (2) edge   node[below=0.5cm ] {\large $a,(3_+,6_+)\vee(0_+,1_+,2_+)$} (4)
    (1) edge   node[above left] {\large $b,(1,2)$} (2);
    \end{tikzpicture}
    \textbf{c}
    }
    \end{minipage}
    \begin{minipage}{0.48\linewidth}
    \centering
    \resizebox{0.9\linewidth}{!}{
    \begin{tikzpicture}[->, >=stealth', shorten >=1pt, auto, node distance=1.5cm, thick,scale=0.8, every node/.style={scale=0.8}, initial where=above]\node[initial, state] (0) {\large $q_0$};
    \node[state] (1) [right = 1.2cm of 0] {\large $q_1$};
    \node[state] (2) [below = 1cm of 1] {\large $q_2$};
    \node[state,accepting] (3) [right = 1.2cm of 1] {\large $q_3$};
    \node[rectangle,draw,minimum width =0.8cm,minimum height = 0.8cm] (4) [below = 1cm of 3] {\large $q_4$};
    \path (0) edge node[above] {\large $a,(1)$} (1)
    (0) edge   node[below left ] {\large $b,(3)$} (2)
    (2) edge   node[right ] {\large $a,((1,3),(2,6),3_+)$} (3)
    (2) edge   node[below=0.5cm ] {\large $a,(0_+\cup3_+,1_+\cup6_+,2_+)$} (4)
    (1) edge   node[above left] {\large $b,(1,2)$} (2);
    \end{tikzpicture}
    \textbf{d}
    }
    \end{minipage}
    \caption{\textbf{a}: the original tAPTA constructed from sELs; \textbf{b}: the leaf nodes are merged; \textbf{c}: all nodes are merged; \textbf{d}: intervals are over-approximated}
    \label{fig:ExMerge}
\end{figure}

\begin{mylemma} \label{lemma:Intoverapp}
The interval over-approximation operation is "safe", i.e., no conflicts are introduced between the timed languages of accepted and rejected states. 
\end{mylemma}

The proof is straightforward from the definition.

\begin{algorithm}[!t]\label{alg:process}
    \small
	\caption{Preprocessing and Simplifying Examples }
	\SetKwInOut{Input}{input}
	\SetKwInOut{Output}{output}
	\Input{A sample set $\Omega=(\Omega_+,\Omega_-)$, alphabet $\Sigma$.}
	\Output{A simplified tAPTA $\mathcal{P}$ recognizing $\Omega$.}
	$\mathit{S_+}\gets \emptyset$; $\mathit{S_-}\gets \emptyset$\;
    \ForEach{$\omega\in\Omega_+$}{$\mathit{S_+} \gets$ \textit{to\_sEL}($\omega$)\;} 
    \ForEach{$\omega\in\Omega_-$}{$\mathit{S_-} \gets$ \textit{to\_sEL}($\omega$)\;} 
    $S_+ \gets \textit{sort\_Lexical}(S_+)$ \tcp*{Sort in lexical order}
    $S_- \gets \textit{sort\_Lexical}(S_-)$\;
    $\mathcal{P}\gets (Q=\{\varepsilon\},q_0=\varepsilon,\Sigma,\delta=\emptyset,F=\emptyset,R=\emptyset)$ \tcp*{Initialize tAPTA}
    $Cevents\gets \varepsilon$  \tcp*{Record the lexical order}
    \ForEach{$s\in S_+\cap S_-$}{
        $Q\gets Q\cup \mathbf{pre}(s)$; $\delta\gets \delta\cup \textit{get\_Transition}(s)$ \tcp*{Add new sEL}
        \eIf{$s\in S_+$}{$F\gets F\cup \{s\}$}{$R\gets R\cup \{s\}$}
        \If{$Cevents\neq \mu(s)$}{
            $Cevents\gets \mu(s)$\;
            $\mathcal{P}\gets \textit{exMerge}(\mathcal{P})$\tcp*{Do extended merging}
            $\delta\gets \textit{int\_Overapp}(\delta)$\tcp*{Do interval over-approximation}
            }
        }
%	\While{$\mathit{flag}=\bot$}{
%            \ForEach{ $\omega\in\Omega_+$}{\label{alg1_line:enumerate}
%                $\xi \gets$ encode($\phi_+,\Omega_-$) \label{alg1_line:encode} \tcp*{Encode $k$-length $p$TRE$_+$ $\phi_+$ once get it.} 
%                $\mathit{flag}, \varphi \gets$ SMTsolver($\xi,\phi_+$)\;
 %               \If{$\mathit{flag}=\top$}{ \label{alg1_line:solution}
 %               \Return $\varphi$\;
 %               }
 %           }
  %          $k\gets k+1$; \label{alg1_line:len_increase}
%        }
%\vspace{-0.5cm}
\end{algorithm}

We denote the simplified tAPTA constructed from sELs $S=(S_+,S_-)$ by $\mathcal{P}(S)=(\mathit{Q},\varepsilon,\Sigma,\delta,F,R)$. $P(S)=(P_+(S),P_-(S))$ is the language pair of $\mathcal{P}(S)$, where $P_+(S)=\bigcup_{q\in F}L(q)$, $P_-(S)=\bigcup_{q\in R}L(q)$. Additionally, we denote the states in the simplified tAPTA that are accessible with multiple paths by set $\mathit{M}$
and let $\mathit{N}=Q\backslash M$ be the set of states that are only accessible with a unique path.

%\begin{mytheorem}
%Given example set $\Omega$,    $\mathcal{P}(S(\Omega))$ 
%\end{mytheorem}

\begin{mytheorem} \label{thm:tAPTA}
$P(S)$ is a safe over-approximation of $S$.
\end{mytheorem}

\begin{proof}
By Definition \ref{def:tAPTA}, the language pair of a  tAPTA $\mathcal{P}_0(S)$ constructed from sELs $S$ is $P_0(S)=S$. By Lemma \ref{lemma:ExMerge} and \ref{lemma:Intoverapp}, the language pair of a simplified tAPTA is a safe over-approximation of the language pair of the original tAPTA. Hence $P(S)$ is a safe over-approximation of $S$.
\end{proof}

%Suppose we do have a na\"ive solution $\varphi=(\bigvee_{\sigma\in\Sigma}{\underline{\sigma}})^*$ that accepts any timed words. However, it is not a meaningful expression with respect to the common understanding.

\begin{myremark}
When constructing an APTA, if the examples set $S=\{S_+, S_-\}$ are input in the lexicographic order, then after inputting sequence $u$, for any previously input sequence $v$, if the longest common prefix of $u,v$ is $p(u,v)$, all states of $\mathbf{pre}(v)\backslash\mathbf{pre}(p(u,v))$ are guaranteed to not have any descendant in further construction (\cite{10.1007/978-3-031-71177-0_4}). With this lexicographic order input, dynamic merging can be applied after each sequence input to save space to store the whole APTA during the construction. Likewise, in the tAPTA construction, we can input sEL sequences in the "lexicographic order" of their event sequences: $S_1,S_2\in\mathcal{SE}, S_1\prec S_2$ if $\mu(S_1)\prec_{lex} \mu(S_2)$. The sEL sequences with the same event sequences are treated as a group and input together. Dynamic extended merging can also be applied. In this way, the number of states of the tAPTA is controlled during the construction.
\end{myremark}

Algorithm \ref{alg:process} outlines the steps of the preprocessing and simplifying examples presented in Section \ref{sbsc:preproc} and \ref{sbsc:tAPTA}.

%% file: SMTmethod.tex
\subsection{Encoding TA Constraints} \label{sbsc:encode}

Upon constructing and simplifying the tAPTA, we proceed to encode it into an SMT formula, ensuring that a solution this formula yields a model corresponding to a timed automaton of a given size. To avoid ambiguity, we hereafter refer to the states and transitions in the tAPTA as "locations" and "edges," respectively, to distinguish them from the states and transitions of the target DTA.

Given a tAPTA $\mathcal{P}=(L,\varepsilon,\delta_{\mathcal{P}},\Sigma,F_{\mathcal{P}},R_{\mathcal{P}})$, we aim to synthesize a DTA $\mathcal{A}=(Q,q_1, \mathcal{C},\Sigma,\delta, F)$ with a specified number of $|Q| = n$ and $|\mathcal{C}| = m$. We map each location $l\in L$ to some state $q\in Q$, and each tAPTA edge to a set of DTA transitions connecting the corresponding states. %A parameter $E$ is assigned to limit the maximum number of DTA transitions with the same source state, target state, event and guard clock. Thus the maximum number of transitions is $\Delta =n^2mE|\Sigma|$. 
We encode all the intrinsic constraints and mappings (of both states and transitions) into an SMT problem $\Phi(\mathcal{P},n,m)$. If $\Phi(\mathcal{P},n,m)$ is UNSAT, we try a larger $(n,m)$. If it is SAT, we construct a DTA ${\mathcal{A}}$ according to the model and return it.

\subsubsection{Basic setting}

\textbf{Forms and ranges.} To facilitate efficient encoding, we restrict the guard form to a conjunction of atomic constraints, where each conjunct is of the form $[g_l,g_u], g_l=g_u$ or $[g_l,g_u),g_l<g_u$, $g_l\in \mathbb{N}$ and $g_u\in \mathbb{N}^+\cup\{\infty\}$. Each guard comprises exactly $|\mathcal{C}|$ conjuncts, one for each clock variable. Since SMT solvers typically do not support encoding $\infty$, we set a sufficiently large maximum upper bound $\kappa$. If the solver returns a synthesized upper bound $g_u$ exceeding $\kappa$, we treat it as $+\infty$ in the learned DTA. We also limit the maximum number of DTA transitions sharing the same source state, target state, and event by a parameter $E$. Thus, the maximum number of transitions is $D =n^2E|\Sigma|$.   We denote the integer range $[1,\cdots,n]$, $[1,\cdots,|L|]$, $[1,\cdots,m]$ and $[1,\cdots,D]$ as $\mathbf{Q}$, $\mathbf{L}$, $\mathbf{C}$ and $\mathbf{D}$, respectively.

\textbf{Variables.} We index the set of states $Q=\{q_1,q_2,\cdots,q_{n}\}$, locations $L=\{l_1,\cdots, l_{|L|}\}$ and clocks $\mathcal{C}=\{c_1,c_2,\cdots,c_m\}$ by the integer sets $\{1,2,\cdots,n\}$, $\{1,\cdots,|L|\}$ and $\{1,2,\\\cdots,m\}$, respectively. Each state $q$ is associated with a Boolean label $f(q)\in\mathbb{B}$, where $f(q)=1\iff q \in F$. Each location is labeled by $h(l)\in \{-1,0,1\}$, where $h(l)=1\iff l \in F_{\mathcal{P}}$, $h(l)=-1\iff l \in R_{\mathcal{P}}$. Each edge $e$ in the tAPTA is represented as a tuple $(prt(e),chd(e),branch(e))$, where $prt(e)\in \mathbf{L}$ is the parent, $chd(e)\in \mathbf{L}$ is the child, and $branch(e)$ represents the set of interval lists forming the disjunction branches on the edge.

A transition $t\in\delta$ is  modeled as a $7-$tuple of variables: 
\begin{itemize}
    \item source state $src(t)\in [1,\cdots,n]$ and target state $tgt(t)\in [1,\cdots,n]$,
    \item event $\sigma(t)\in [1,\cdots,|\Sigma|]$ and index $ind(t)\in [1,\cdots,E]$,
    \item clock reset vector $r(t) \in \mathbb{B}^m$, guard lower bounds $g_l(t)\in \mathbb{N}^m$ and guard upper bounds $g_u(t)\in (\mathbb{N}^+\cup\{\infty\})^m$.
\end{itemize}
The $c$-th components $g_l(t)(c)$ and $g_u(t)(c)$  imply that the corresponding conjunct for the $c$-th clock in the guard is either $[g_l(t)(c),g_u(t)(c))$ or $[g_l(t)(c),g_u(t)(c)]$. 

\textbf{Constraints.} We consider two categories of constraints for the encoding: (1) Intrinsic constraints enforce the structural properties of the DTA, including determinism and transition maximum. (2) Language constraints ensure behavioral consistency with the tAPTA, including the mappings from locations to states, and the feasibility of transitions over edges.

\subsubsection{Intrinsic Constraints}

We first encode the intrinsic constraints that are independent of the specific structure of the tAPTA.
\begin{itemize}
    \item \textit{Determinism.} If two transitions have the same source state and event, their guards must be disjoint.
   $$ \forall i,j\in \mathbf{D}, i\neq j: src(i)=src(j)\wedge \sigma(i)=\sigma(j)\Longrightarrow $$
   $$ \underset{c\in \mathbf{C}}{\bigvee}g_l(i)(c)>g_u(j)(c)\vee g_l(j)(c)>g_u(i)(c)\vee g_l(i)(c)=$$
   $$g_u(j)(c)>g_l(j)(c) \vee g_l(j)(c)=g_u(i)(c)>g_l(i)(c)$$

    \item \textit{Transition maximum.} The number of transitions with the same source state, target state, and event is limited.
$$\forall i,j\in \mathbf{D}, i\neq j: src(i)=src(j)\wedge tgt(i)=tgt(j) \wedge$$ 
$$\sigma(i)=\sigma(j)\Longrightarrow ind(i)\neq ind(j)$$
\end{itemize}

\subsubsection{Mapping Constraints}
We use accessibility variables $d_{i,j}, i\in\mathbf{L}, j\in \mathbf{Q}$ to represent the mapping from locations to states: 
\[
d_{i,j} = 
\begin{cases} 
    1, & \text{if } \exists \mathcal{L}\in \mathcal{EL}, \, \mathcal{L}\in \mathcal{EL}(l_i)\cap\mathcal{EL}(q_j)\\
    0, & \text{otherwise } 
\end{cases}
\] 

The variables indicates that a subset of the timed language of a location also reaches a state in the DTA.

\begin{itemize}

\item \textit{Mapping consistency.} If a state is mapped to by an accepting location, it must be an accepting state. If a state is mapped to by a rejecting location, it must not be an accepting state.
$$\forall i\in \mathbf{L}, j\in \mathbf{Q} \;s.t. d_{i,j}=1:$$ 
$$h(i)=-1 \Longrightarrow f(j) = 0, h(i)=1 \Longrightarrow f(j) =1$$

\item \textit{Initial state.} The initial location $\epsilon$ is  mapped to the initial state $q_1$. 
$$d_{1,1}=1$$

\item  \textit{Unique mapping.} Locations that are only accessible with a unique path cannot be mapped to two or more different states.
$$\forall i\in \mathbf{L}, l_i\in N, j,j'\in \mathbf{Q}, j\neq j' : \lnot (d_{i,j}\wedge d_{i,j'})$$

\end{itemize}

\subsubsection{Feasibility Constraints}
We introduce a vector $v:Path\times\mathbb{N}\rightarrow\mathbb{N}^m$ to record the number of edges in the path that have been traversed since the last reset of each clock. Initially $v(p,1)=\mathbf{1}$. 

We enumerate all paths leading to a leaf location, and traverse each path $p$ from the root $\varepsilon$ to the leaf to encode the feasibility of the transitions over the edges in the path. When traversing the $i$-th edge in the path $p(i)$, we first decide which branch of the disjunction of the interval lists is taken if the branches haven't been over-approximated. Let $list(p,i)$ be the branch decided.
If more than one branch can be taken, we encode that at least one of the branches passes a transition.

After deciding the interval list, for the $j$-th clock, the value of the clock at this step falls in the  $v(p,i)_j$-th interval in the list. We use a function $pass():Interval \times \mathbb{N}\times\mathbb{N}\rightarrow \mathbb{B}$ to indicate whether an union of interval $I$ has intersection with the guard interval $[g_l,g_u)$ (or $[g_l,g_u]$): 
$$ pass(I,g_l,g_u)=
\begin{cases} 
    1, & \text{if } (g_l=g_u\wedge g_l\in I) \vee \\
    & (g_l<g_u\wedge[g_l,g_u)\cap I \neq \emptyset)\\
    0, & \text{otherwise } 
\end{cases}$$

If the interval has an intersection with the guard interval, $v(p,i+1)$ is valued according to the transition's reset: if a clock $j\in \mathbf{C} $ is reset, $v(p,i+1)_j=1$, otherwise $v(p,i+1)_j=v(p,i)_j+1$.

The encoding of the feasibility of a transition in the DTA over an edge of the tAPTA is: 
$$\forall p,\,i\in [1,length(p)],l,m\in \mathbf{L}, r,s \in  \mathbf{Q}, t \in \mathbf{D}: $$
$$prt(p(i))=l \wedge chd(p(i))=m\wedge d_{l,r}=1 \wedge src(t)=r \wedge$$
$$tgt(t)=s \wedge (\forall j\in \mathbf{C},\, pass(list(p,i)_{(v(p,i)_j)},g_l(t)_j,g_u(t)_j) )$$
$$\Longrightarrow d_{m,s}=1 \wedge (\forall j\in \mathbf{C},\,v(p,i+1)_j=v(p,i)_j\cdot (1-r(t)_j)+1)$$

\subsection{Constructing the Learned TA}

Upon solving $\Phi(\mathcal{P},n,m)$, we obtain values for the state labels, guards, and clock variables to instantiate ${\mathcal{A}}$. We then perform a final simplification step: the synthesized automaton may contain superfluous transitions resulting from arbitrary values assigned by the solver. These transitions are not required to recognize the given examples. We prune these transitions by simulating all positive examples on the learned TA and removing any transitions that are never traversed.

\begin{mylemma}
    \label{lemma:DTA}
Given a tAPTA $\mathcal{P}$ with its language pair $P=(P_+,P_-)$, there exists a DTA $\mathcal{A}(\mathcal{P})$, such that $P_+\subset\mathcal{L}(\mathcal{A}(\mathcal{P}))$ and $P_-\cap\mathcal{L}(\mathcal{A}(\mathcal{P}))=\emptyset$.
\end{mylemma}

\begin{proof}
Given a tAPTA $\mathcal{P}=(\mathit{Q},\varepsilon,\Sigma,\delta,F,R)$, let $P(F)$ be the set of paths incoming to the states in $F$ and the maximum length of these paths is $l_m$. A DTA $\mathcal{A}(\mathcal{P})=(\mathit{Q},\varepsilon,\mathcal{C},\Sigma_A,\delta_A,F)$ is constructed with $\mathcal{C}=(c_1,c_2,\cdots,c_{l_m})$, $\Sigma=\{\sigma\in\Sigma\,|\,\exists\, q_1,q_2, \,\delta(q_1,\sigma)=q_2\}$. And $\delta_A$ is as follows: for each edge $d=(q_{j-1},u,q_j)$ that is the $j-$th edge in a path in $P(F)$, $u$ is the labeling tuple of the edge, where $u=(\sigma,\underset{i=1,\cdots,k}{\bigvee}\Lambda_i)$, each branch of the interval lists is $\Lambda_i=(I_1^i,I_2^i,\cdots,I_j^i)$ (branches that are not of length $j$ are not considered). Then there are $k$ transitions $(q_{j-1},\sigma,g_i,\{c_{j+1}\},q_{j})$ in $\delta_A$, where $i=1,\cdots,k$, $g_i$ is $\underset{n=1,\cdots,j}\bigwedge {c_n\in I_n^i}$. No other transitions are included in $\delta_A$. By construction of $\delta_A$, $P_+=\mathcal{L}(\mathcal{A}(\mathcal{P}))$. Since the timed language of tAPTA $P_+\cap P_- = \emptyset$, we have $P_-\cap\mathcal{L}(\mathcal{A}(\mathcal{P}))=\emptyset$. 
\end{proof}

%% file: guarantees.tex
\subsection{Soundness and Termination}
In this section, we establish the theoretical guarantees of our proposed framework, including the properties of soundness and termination.

\begin{mytheorem}[Soundness w.r.t Problem \ref{prob:DTAmining}] \label{thm:overallcorrect}
The DTA $\mathcal{A}$ mined from example set $\Omega=(\Omega_+,\Omega_-)$ is a solution to Problem \ref{prob:DTAmining}. 
\end{mytheorem}

\begin{proof}
Suppose $\mathcal{P}$ is the tAPTA that the SMT formula is encoded from. The \textbf{Mapping Constraints} ensure that each accepting location in $\mathcal{P}$ is mapped exclusively to accepting states in  $\mathcal{A}$, while rejecting locations in $\mathcal{P}$ are not mapped to any accepting states. The \textbf{Feasibility Constraints} ensure that each timed word of a location reaches one of the states the location is mapped to. Consequently, $P_+\subset\mathcal{L}(\mathcal{A}(\mathcal{P}))$ and $P_-\cap\mathcal{L}(\mathcal{A}(\mathcal{P}))=\emptyset$. Thus, $\mathcal{A}$ recognizes the language pair of $\mathcal{P}$. By Theorem \ref{thm:tAPTA}, $\mathcal{A}$ recognizes the sELs from which $\mathcal{P}$ is constructed. By Theorem \ref{thm:sEL}, $\mathcal{A}$ also recognizes the original example set $\Omega$. Hence, $\mathcal{A}$ is a solution to Problem \ref{prob:DTAmining}. 
\end{proof}

\begin{mytheorem}[Termination w.r.t Problem \ref{prob:DTAmining}]\label{thm:overallterminate}
Given the set $\Omega=(\Omega_+,\Omega_-)$, the  proposed DTA mining approach is guaranteed to terminate and return a solution to Problem \ref{prob:DTAmining} in finite steps under a specific enumeration order for the DTA size parameters $(n,m)$.
\end{mytheorem}

\begin{proof}
By Definition \ref{def:sel}, \ref{def:tAPTA}, \ref{def:ExMerge} and \ref{def:Intoverapp}, the preprocessing of the examples, and the construction and the simplification of tAPTA are finite processes. By Lemma 6, there exists a DTA $\mathcal{A}$ of size $(n,m)=(|Q|,l_m)$ that is a solution to Problem \ref{prob:DTAmining}, where $|Q|$ and $l_m$ relate to the simplified tAPTA. By Definition \ref{def:sel},\ref{def:ExMerge} and \ref{def:Intoverapp}, we have $|Q|\leq \underset{\omega\in\Omega_+\cup\Omega_-}{\Sigma}|\omega|+1$ and $l_m\leq \underset{\omega\in\Omega_+\cup\Omega_-}{max}|\omega|+1$. Thus, there exists an enumeration order of DTA sizes that is guaranteed to reach $(|Q|,l_m)$ in finite steps. Hence, the overall approach terminates and returns a solution to Problem \ref{prob:DTAmining} in finite steps. 
\end{proof}

\begin{mytheorem}[Termination w.r.t Problem \ref{prob:mDTAmining}]\label{thm:minDTAproperty}
Given the set $\Omega=(\Omega_+,\Omega_-)$, if simplification (Def.~\ref{def:ExMerge}, \ref{def:Intoverapp}) are not applied and the SMT formula is directly encoded from the unsimplified tAPTA, then the  proposed DTA mining approach is guaranteed to terminate and return a solution to Problem \ref{prob:mDTAmining} in finite steps under a specific enumeration order for the DTA size parameters $(n,m)$.
\end{mytheorem}

\begin{proof}
Suppose the unsimplified tAPTA is $\mathcal{P}_0$ and the sEL set it is constructed from is $S$, and a resulting DTA mined is $\mathcal{A}$. Similar to the proof of Theorem \ref{thm:overallterminate}, we can prove such $\mathcal{A}$ exists and can be returned in finite steps. Similar to the proof of Theorem \ref{thm:overallcorrect}, $\mathcal{A}$ recognizes $P_0(S)$. By Definition \ref{def:tAPTA}, the language pair of $\mathcal{P}_0(S)$ is $P_0(S)=S$. Thus, $\mathcal{A}\in A(S)$. By Theorem \ref{thm:sEL}, $\mathcal{A}\in A(\Omega)$. Thus, $D(\Omega)\neq\emptyset$.

$\Omega$ is finite and each trace $\omega$ is finite. Thus, finitely many clocks are sufficient to characterize all possible sum values of the delay times of the timed words in $\Omega$. Let $m_0$ denote this number of clocks. Consequently, there exist $\mathcal{D}_0\in D(\Omega)$, such that its size is $(n_0,m_1)$, $m_1\leq m_0$. Consider the enumeration order of (n,m) as follows: $(1,1)\rightarrow (1,2)\rightarrow\cdots\rightarrow(1,m_0)\rightarrow(2,1)\rightarrow\cdots$. Under this enumeration order, the search is guaranteed to terminate at some pair $(n_0,m_2)$ with $m_2\leq m_1$, returning a solution to Problem \ref{prob:mDTAmining} in finite steps.
\end{proof}

%% file: experiments.tex
\section{Comparison with Existing Work} \label{sc:compare}
This section presents a comparative analysis between our approach and the SMT-based DTA synthesis method proposed by Tappler et al. \cite{TapplerAL22}. While both approaches utilize SMT solving to synthesize a DTA, they differ significantly in data formats, modeling assumptions, and encoding strategies. 

\textbf{Difference in Problem Setting:} Tappler's approach uses input-and-output traces, where each transition is labeled with either an input or an output event. In contrast, our approach operates exclusively on input traces.  Tappler's model distinguishes delay transitions and discrete transitions, whereas our model semantically integrates both delay and discrete transitions into a single type of transition.

\textbf{Negative Traces:} Tappler's approach relies exclusively on positive samples,  based on the premise that real-world system testing typically observes only valid executions (i.e., positive traces) without access to negative ones. Consequently, their model requires location invariants and the $\mathit{k-Urgency}$ assumption to prevent overgeneralization. Our approach exploits both positive and negative traces, thereby eliminating the need for additional assumptions regarding delay times. 

\textbf{Example Merging:} Tappler's approach encodes SMT formulas directly from traces, which ignores the redundancy within the example sets. Since their approach does not process negative examples, it does not account for conflicts. In our approach, traces are first translated into sEL, and then are over-approximated and simplified via the tAPTA. This process not only detects conflicts and eliminates redundancy but also merges the timed languages within the tAPTA, significantly reducing the number of clauses in the SMT encoding. 

\textbf{Guards and Clock Constraints:} Tappler's approach initially limits the encoding in one-clock case and then duplicates the constraints when additional clocks are added. This renders the clocks "independent", meaning that each transition’s guard contains at most one inequality over a unique clock $c \in \mathcal{C}$ and can reset only $c$. This imposes a significant restriction on the expressiveness of the resulting DTA.  In contrast, our approach accommodates normal clock constraints as defined in Definition \ref{def:ta}.

\section{Implementation and Experiments} \label{sc:experiments}
\subsection{Theory and Encoding Complexity}\label{sbsc:implement}

Data types in the constraints listed in Sect. \ref{sbsc:encode} comprise: (1) discrete structural data for the events, indices of locations, states, edges, etc; (2) integer bounds in guards and the intervals of accumulated delay times in edges. Since the real-valued delay times and clock values are replaced by integral-bounded intervals after translating the traces to sELs, integer theory is sufficient for our encoding.

\textit{Encoding Complexity.} Regarding the number of constraint clauses in the SMT encoding, the dominant factors are: (i) $\mathcal{O}(m\cdot D)$ from determinism, (ii) $\mathcal{O}(D)$ from transition maximum, (iii) $\mathcal{O}(n\cdot |L|)$ from mapping consistency, (iv) $\mathcal{O}(n^2\cdot |L|^2)$ from uniqueness, (v) $\mathcal{O}(mn^2\cdot |L|^2 \cdot D)$ from accessibility.
Since $D = n^2E|\Sigma|$, the overall number of constraints is $\mathcal{O}(mn^4\cdot |\Sigma | \cdot E\cdot |L|^2 \cdot D)$.
%If no redundancy is discarded in trace preprocessing - meaning that each trace has a corresponding sEL translated, and no interval over-approximation can be done in the simplification of tAPTA - meaning that each sEL is kept in the interval lists branches of a path without being merged with another sEL, then the number of branches we need to encode is equal to the number of original traces. %Let $b(e)$ be the number of interval lists branches of an edge $e$ ($b(e)=1$ if it is not of the form of the disjunction of interval lists), $B(p)=\underset{e\, in\, p}{\Sigma}b(e)$ is the number of branches in total of a path $p$.

\subsection{Experiments}

In this section, we present three groups of experiments to evaluate our approach. The first group compares the number of sELs, tAPTA locations and the constraint clauses encoded by our approach and Tappler's approach, using randomly generated positive example sets. The second group assesses the scalability of our approach by varying numbers of positive and negative input traces. The third group is a case study on learning a task model by mining a DTA for the scheduling system. The program is developed in Python 3.8. We use Z3~\cite{MouraB08} as our SMT solver. All experiments were conducted on a laptop:
\begin{itemize}
    \item Processor: 11th Gen Intel(R) Core(TM) i5-11300H @ 3.10GHz, with 4 cores and 8 logical processors.
    \item RAM: 16GB installed
\end{itemize}

\subsubsection{Numbers of sELs, tAPTA Locations and Constraints Encoded}

This experiment demonstrates the capability of our approach to simplify massive input timed traces for SMT-based DTA mining. We randomly generated a series of timed trace sets and processed them using both Tappler's direct encoding (denoted as SMT) and our approach (denoted as Merge-SMT) respectively. To isolate the impact of different guard settings, both approaches were run in one-clock case. The traces were generated over an alphabet $\Sigma=\{a,b,c\}$ with each delay time randomly distributed in $(0,5]$ and trace lengths ranging from 4 to 8. The results are summarized in Table \ref{tab:constrts}.

\begin{table}[tb]
\caption{Numbers of samples, sELs, tAPTA locations, and constraints between different approaches across sample sizes. }
\begin{center}
\begin{tabular}{|c|c|c|c|c|}
\hline
\textbf{Traces} & 10 & 25 & 50 & 100  \\
\hline
\textbf{sELs} & $10$ & $25$ & $50$ & $100$\\
\hline
\textbf{tAPTA locations} & $36$ & $59$ &$91$ &$157$\\
\hline
\textbf{SMT cons} & $221$ & $527$ &$1236$& 2363 \\
\hline
\textbf{Merge-SMT cons} & $272$ & $614$ &$1153$& 1706 \\
\hline
\end{tabular}

\vspace{0.3cm}
\begin{tabular}{|c|c|c|c|c|}
\hline
\textbf{Traces} &  200 & 300 & 500 &750 \\
\hline
\textbf{sELs} &  $200$ & 300 & 500 &750\\
\hline
\textbf{tAPTA locations} & $238$ & 306 & 449 &$593$\\
\hline
\textbf{SMT cons} &  4706 &$6917$& $11859$ & $17642$\\
\hline
\textbf{Merge-SMT cons}  & 2858 &$4094$& $5737$ & $7521$\\
\hline
\end{tabular}
\label{tab:constrts}
\end{center}
 \vspace{-0.5cm}
\end{table}

No redundant traces were discarded during the preprocessing phase (translating traces into sELs), as explained in Remark \ref{rmk:sEL}. However, the construction and simplification of the tAPTA significantly reduced the number of locations to be encoded: when the number of samples is moderate, the number of locations in an unsimplified tAPTA is close to the sum of the lengths of all traces. After merging, noting that negative traces were not considered, the tAPTA can be simplified to a structure identical to that of simplified APTA constructed from untimed positive traces. This simplification via tAPTA substantially reduced encoding complexity as the number of samples increases.

\subsubsection{Scalability Tests}

This experiment evaluates the efficiency of our approach in mining a DTA from given positive and negative traces. We randomly generated a series of target DTAs featuring 2 to 6 states, 1 to 3 clocks, and 2 to 4 distinct events. By sampling the DTAs with uniform sampling tool \textbf{WORDGEN} \cite{Barbot2023WORDGEN}, we produced positive trace sets ranging from 50 to 600 traces. Corresponding negative trace sets of equal size were sampled from the complement DTAs of the target DTAs. The maximum upper bounds of guards in the target DTAs and the complemented DTA is $\kappa=10$, and the transition limit is $E=1$. For each target DTA and each trace set size, we conducted 5 trials. The time limit of each trial is 900 seconds. The results are summarized in Table~\ref{tab:scal1} $\sim$ \ref{tab:scal3}.

\begin{table}[!h]
\caption{Average time expanse (seconds) and the number of successes over various sample sizes, 1-clock target DTA. $\textbf{Above}$: $|\Sigma|=2$ . $\textbf{Middle}$:$|\Sigma|=3$. $\textbf{Below}$: $|\Sigma|=4$.}
\vspace{-0.2cm}
\label{tab:scal1}
\begin{center}
\begin{tabular}{|c|c|c|c|c|c|}
\hline
\textbf{Number}&\multicolumn{5}{c|}{\textbf{Number of states $n$}} \\
\cline{2-6} 
\textbf{of traces} & 2 & 3 &4 &5 &6 \\
\hline
50& 4.2/5 & 10.1/5 & 23.3/5 & $24.2/0$ & $27.6/0$\\
\hline
100& 6.4/5 & 13.6/5 & 27.5/5 & 68.8/5 & $92.3/2$\\
\hline
300& 7.9/5 & 18.4/5 & 32.2/5 &76.0/5& $148/4$ \\
\hline
600& 10.1/5 & 22.5/5 & 35.8/5 & 82.3/5& 161/5\\
\hline
\end{tabular}

\vspace{0.5cm}

\begin{tabular}{|c|c|c|c|c|c|}
\hline
\textbf{Number}&\multicolumn{5}{c|}{\textbf{Number of states $n$}} \\
\cline{2-6} 
\textbf{of traces} & 2 & 3 &4 &5 &6 \\
\hline
50& 7.3/5 & 17.0/5 & $29.6/2$ & $28.7/0$ & $30.3/0$\\
\hline
100& 11.4/5 & 24.5/5 & 51.0/5 & $68.8/1$ & $76.1/0$\\
\hline
300& 16.2/5 & 30.8/5 & 62.9/5 &142/5& $291/3$ \\
\hline
600& 22.7/5 & 38.4/5 & 73.5/5 & 169/5 & 352/5 \\
\hline
\end{tabular}

\vspace{0.5cm}
\begin{tabular}{|c|c|c|c|c|c|}
\hline
\textbf{Number}&\multicolumn{5}{c|}{\textbf{Number of states $n$}} \\
\cline{2-6} 
\textbf{of traces} & 2 & 3 &4 &5 &6 \\
\hline
50& 9.5/5 & $19.2/2$ & $23.1/0$ & $21.9/0$ & $26.6/0$\\
\hline
100& 16.7/5 & 36.0/5 & $65.4/2$ & $77.3/0$ & $68.5/0$\\
\hline
300& 21.6/5 & 48.6/5 & 102/5 &$131/0$& $147/0$ \\
\hline
600& 27.8/5 & 63.9/5 & 125/5 & $193/1$ & $212/0$ \\
\hline
\end{tabular}
\end{center}
 \vspace{-0.5cm}
\end{table}

\begin{table}[!h]
\caption{Average time expanse (seconds) and the number of successes over various sample sizes, 2-clock target DTA. $\textbf{Above}$: $|\Sigma|=2$ . $\textbf{Middle}$:$|\Sigma|=3$. $\textbf{Below}$: $|\Sigma|=4$.}
\vspace{-0.2cm}
\label{tab:scal2}
\begin{center}
\begin{tabular}{|c|c|c|c|c|c|}
\hline
\textbf{Number}&\multicolumn{5}{c|}{\textbf{Number of states $n$}} \\
\cline{2-6} 
\textbf{of traces} & 2 & 3 &4 &5 &6 \\
\hline
50& 6.6/5 & $14.1/5$ & $29.3/2$ & $42.5/0$ & $37.9/0$\\
\hline
100& 9.4/5 & 19.8/5 & $44.2/5$ & $71.3/1$ & $76.0/0$\\
\hline
300& 15.6/5 & 29.7/5 & 64.5/5 &$96.2/0$& $107/0$ \\
\hline
600& 22.8/5 & 46.3/5 & 87.4/5 & 192/5 & $201/0$\\
\hline
\end{tabular}

\vspace{0.5cm}

\begin{tabular}{|c|c|c|c|c|c|}
\hline
\textbf{Number}&\multicolumn{5}{c|}{\textbf{Number of states $n$}} \\
\cline{2-6} 
\textbf{of traces} & 2 & 3 &4 &5 &6 \\
\hline
50& 8.1/5 & $18.3/4$ & $20.6/0$ & $20.3/0$ & $23.7/0$\\
\hline
100& 13.4/5 & 30.5/5 & $54.2/1$ & $56.8/0$ & $45.0/0$\\
\hline
300& 18.8/5 & 42.1/5 & 77.5/5 &$86.2/0$& $109/0$ \\
\hline
600& 26.2/5 & 48.9/5 & 86.7/5 & $136/2$ & $151/0$ \\
\hline
\end{tabular}

\vspace{0.5cm}
\begin{tabular}{|c|c|c|c|c|c|}
\hline
\textbf{Number}&\multicolumn{5}{c|}{\textbf{Number of states $n$}} \\
\cline{2-6} 
\textbf{of traces} & 2 & 3 &4 &5 &6 \\
\hline
50& 11.4/5 & $21.4/2$ & $20.5/0$ & $22.1/0$ & $24.7/0$\\
\hline
100& 17.8/5 & $39.2/4$ & $57.0/0$ & $60.3/0$ & $64.6/0$\\
\hline
300& 25.1/5 & 54.3/5 & $76.1/1$ &$95.8/0$& $79.2/0$ \\
\hline
600& 31.6/5 & 66.4/5 & 139/5 & $152/0$ & $156/0$ \\
\hline
\end{tabular}
\end{center}
% \vspace{-0.5cm}
\end{table}

\begin{table}[!t]
\caption{Average time expanse (seconds) and the number of successes over various sample sizes, 3-clock target DTA. $\textbf{Above}$: $|\Sigma|=2$ . $\textbf{Middle}$:$|\Sigma|=3$. $\textbf{Below}$: $|\Sigma|=4$.}
% \vspace{-0.2cm}
\label{tab:scal3}
\begin{center}
\begin{tabular}{|c|c|c|c|c|c|}
\hline
\textbf{Number}&\multicolumn{5}{c|}{\textbf{Number of states $n$}} \\
\cline{2-6} 
\textbf{of traces} & 2 & 3 &4 &5 &6 \\
\hline
50& 10.4/5 & $12.6/0$ & $13.8/0$ & $12.9/0$ & $14.4/0$\\
\hline
100& 17.1/5 & $32.3/2$ & $36.4/0$ & $43.9/0$ & $42.0/0$\\
\hline
300& 23.3/5 & 43.7/5 & $45.8/0$ &$47.2/0$& $48.5/0$ \\
\hline
600& 28.5/5 & 65.1/5 & $83.2/1$ & $70.6/0$ & $74.2/0$ \\
\hline
\end{tabular}

\vspace{0.5cm}

\begin{tabular}{|c|c|c|c|c|c|}
\hline
\textbf{Number}&\multicolumn{5}{c|}{\textbf{Number of states $n$}} \\
\cline{2-6} 
\textbf{of traces} & 2 & 3 &4 &5 &6 \\
\hline
50& 14.6/5 & $15.7/0$ & $16.1/0$ & $17.3/0$ & $16.9/0$\\
\hline
100& 19.5/5 & $21.5/0$ & $21.8/0$ & $23.0/0$ & $23.6/0$\\
\hline
300& 27.3/5 & $48.2/2$ & $43.7/0$ &$50.5/0$& $45.8/0$ \\
\hline
600& 41.7/5 & 94.8/5 & $114/0$ & $142/0$ & $101/0$ \\
\hline
\end{tabular}

\vspace{0.5cm}
\begin{tabular}{|c|c|c|c|c|c|}
\hline
\textbf{Number}&\multicolumn{5}{c|}{\textbf{Number of states $n$}} \\
\cline{2-6} 
\textbf{of traces} & 2 & 3 &4 &5 &6 \\
\hline
50& $16.2/3$ & $17.4/0$ & $17.3/0$ & $16.9/0$ & $17.6/0$\\
\hline
100& 23.8/5 & $22.1/0$ & $24.0/0$ & $21.9/0$ & $23.6/0$\\
\hline
300& 32.1/5 & $39.7/0$ & $41.3/0$ &$34.7/0$& $43.8/0$ \\
\hline
600& 48.5/5 & $119/4$ & $125/0$ & $121/0$ & $124/0$ \\
\hline
\end{tabular}
\end{center}
\vspace{-0.5cm}
\end{table}

The results indicate that our approach is capable of mining a DTA from hundreds of positive and negative traces within a reasonable time. Notably, when the number of traces far exceeded the minimum required to synthesize a relatively simple DTA, our approach significantly reduced encoding complexity and demonstrated scalability. The number of states $n$ in the target DTA is the primary factor influencing computation time. However, due to the inherent limitations of passive learning, the mined DTA may not be isomorphic to the target DTA when the number of traces is limited or the complexity of the target system increases. The number of clocks $m$ is the most influential factor. In some instances, discrepancies were confined to the guards; however, the mined DTA was more frequently simpler than the target DTA, particularly regarding $m=2,3$. 

This phenomenon can be attributed to the fact that passive learning aims to synthesize the minimal model recognizing the sampled traces. For instance, traces randomly sampled from a target DTA with $n=6$ states might be recognizable by a DTA with only $n=4$ states. The search strategy proceeds to $n=5,6$ only if $n=4$ fails to satisfy the formulas encoded from both positive and negative traces. Conversely, if the negative traces are insufficient, the search may terminate at $n=4$, resulting in a shorter runtime. Note that the purpose of passive learning is to learn a model consistent with the given trace sets, not the exact target model. Therefore, whether a passive learning method returns an exact model depends on the sampled traces, not the method itself. If the sampled traces fully characterize the target model, a passive learning method returns the exact model successfully. How to generate the characterization trace set of a DTA remains an open question. 

We present the consistency results between the learned DTA and the test sets (600 traces) in Tables \ref{tab:qual1}-\ref{tab:qual3}. We randomly generated test traces from the target DTA and tested whether the corresponding learned DTA is consistent with the target one on the acceptance and rejection of the test traces. $"1"$ indicates that the learned DTA is consistent with the target DTA over all test traces.

\begin{table}[h]
\caption{Average consistency (proportion) over various sample sizes, 1-clock target DTA. $\textbf{Above}$: $|\Sigma|=2$ . $\textbf{Middle}$:$|\Sigma|=3$. $\textbf{Below}$: $|\Sigma|=4$.}
\label{tab:qual1}
\begin{center}
\begin{tabular}{|c|c|c|c|c|c|}
\hline
\textbf{Number}&\multicolumn{5}{c|}{\textbf{Number of states $n$}} \\
\cline{2-6} 
\textbf{of traces} & 2 & 3 &4 &5 &6 \\
\hline
50& 1 & 1 & 1 & $0.922$ & $0.914$\\
\hline
100& 1 & 1 & 1 & 1 & $0.956$\\
\hline
300& 1 & 1 & 1 & 1 & $0.988$ \\
\hline
600& 1 & 1 & 1 & 1 & 1 \\
\hline
\end{tabular}

\vspace{0.2cm}

\begin{tabular}{|c|c|c|c|c|c|}
\hline
\textbf{Number}&\multicolumn{5}{c|}{\textbf{Number of states $n$}} \\
\cline{2-6} 
\textbf{of traces} & 2 & 3 &4 &5 &6 \\
\hline
50& 1 & 1 & $0.917$ & $0.905$ & $0.894$\\
\hline
100& 1 & 1 & 1 & $0.952$ & $0.941$\\
\hline
300& 1 & 1 & 1 & 1 & $0.987$ \\
\hline
600& 1 & 1 & 1 & 1 & 1 \\
\hline
\end{tabular}

\vspace{0.2cm}
\begin{tabular}{|c|c|c|c|c|c|}
\hline
\textbf{Number}&\multicolumn{5}{c|}{\textbf{Number of states $n$}} \\
\cline{2-6} 
\textbf{of traces} & 2 & 3 &4 &5 &6 \\
\hline
50& 1 & $0.914$ & $0.906$ & $0.892$ & $0.885$\\
\hline
100& 1 & 1 & $0.954$ & $0.942$ & $0.937$\\
\hline
300& 1 & 1 & 1 &$0.991$& $0.989$ \\
\hline
600& 1 & 1 & 1 & $0.995$ & $0.993$ \\
\hline
\end{tabular}
\end{center}
 \vspace{-0.5cm}
\end{table}

\begin{table}[!h]
\caption{Average consistency (proportion) over various sample sizes, 2-clock target DTA. $\textbf{Above}$: $|\Sigma|=2$ . $\textbf{Middle}$:$|\Sigma|=3$. $\textbf{Below}$: $|\Sigma|=4$.}
\label{tab:qual2}
\begin{center}
\begin{tabular}{|c|c|c|c|c|c|}
\hline
\textbf{Number}&\multicolumn{5}{c|}{\textbf{Number of states $n$}} \\
\cline{2-6} 
\textbf{of traces} & 2 & 3 &4 &5 &6 \\
\hline
50& 1 & $ 1 $ & $0.912$ & $0.904$ & $0.892$\\
\hline
100& 1 & 1 & 1 & $0.953$ & $0.941$\\
\hline
300& 1 & 1 & 1 &$0.992$& $0.989$ \\
\hline
600& 1 & 1 & 1 & 1 & $0.996$\\
\hline
\end{tabular}

\vspace{0.2cm}

\begin{tabular}{|c|c|c|c|c|c|}
\hline
\textbf{Number}&\multicolumn{5}{c|}{\textbf{Number of states $n$}} \\
\cline{2-6} 
\textbf{of traces} & 2 & 3 &4 &5 &6 \\
\hline
50& 1 & $0.908$ & $0.901$ & $0.894$ & $0.890$\\
\hline
100& 1 & 1 & $0.947$ & $0.941$ & $0.935$\\
\hline
300& 1 & 1 & 1 &$0.990$& $0.986$ \\
\hline
600& 1 & 1 & 1 & $0.995$ & $0.993$ \\
\hline
\end{tabular}

\vspace{0.2cm}
\begin{tabular}{|c|c|c|c|c|c|}
\hline
\textbf{Number}&\multicolumn{5}{c|}{\textbf{Number of states $n$}} \\
\cline{2-6} 
\textbf{of traces} & 2 & 3 &4 &5 &6 \\
\hline
50& 1 & $0.903$ & $0.896$ & $0.891$ & $0.884$\\
\hline
100& 1 & $0.952$ & $0.944$ & $0.940$ & $0.933$\\
\hline
300& 1 & 1 & $0.989$ &$0.986$& $0.982$ \\
\hline
600& 1 & 1 & 1 & $0.993$ & $0.992$ \\
\hline
\end{tabular}
\end{center}
\end{table}

\begin{table}[!h]
\caption{Average consistency (proportion) over various sample sizes, 3-clock target DTA. $\textbf{Above}$: $|\Sigma|=2$ . $\textbf{Middle}$:$|\Sigma|=3$. $\textbf{Below}$: $|\Sigma|=4$.}
\label{tab:qual3}
\begin{center}
\begin{tabular}{|c|c|c|c|c|c|}
\hline
\textbf{Number}&\multicolumn{5}{c|}{\textbf{Number of states $n$}} \\
\cline{2-6} 
\textbf{of traces} & 2 & 3 &4 &5 &6 \\
\hline
50& 1 & $0.898$ & $0.892$ & $0.887$ & $0.881$\\
\hline
100& 1 & $0.946$ & $0.941$ & $0.935$ & $0.932$\\
\hline
300& 1 & 1 & $0.987$ &$0.984$& $0.982$ \\
\hline
600& 1 & 1 & $0.994$ & $0.992$ & $0.991$ \\
\hline
\end{tabular}

\vspace{0.2cm}

\begin{tabular}{|c|c|c|c|c|c|}
\hline
\textbf{Number}&\multicolumn{5}{c|}{\textbf{Number of states $n$}} \\
\cline{2-6} 
\textbf{of traces} & 2 & 3 &4 &5 &6 \\
\hline
50& 1 & $0.892$ & $0.885$ & $0.876$ & $0.863$\\
\hline
100& 1 & $0.939$ & $0.936$ & $0.931$ & $0.924$\\
\hline
300& 1 & $0.981$ & $0.979$ &$0.976$& $0.974$ \\
\hline
600& 1 & 1 & 0.993 & 0.990 & 0.988 \\
\hline
\end{tabular}

\vspace{0.2cm}
\begin{tabular}{|c|c|c|c|c|c|}
\hline
\textbf{Number}&\multicolumn{5}{c|}{\textbf{Number of states $n$}} \\
\cline{2-6} 
\textbf{of traces} & 2 & 3 &4 &5 &6 \\
\hline
50& $0.911$ & $0.887$ & $0.878$ & $0.864$ & $0.855$\\
\hline
100& 1 & $0.932$ & $0.928$ & $0.923$ & $0.919$\\
\hline
300& 1 & $0.980$ & $0.977$ &$0.975$& $0.972$ \\
\hline
600& 1 & 0.992 & $0.990$ & $0.989$ & $0.986$ \\
\hline
\end{tabular}
\end{center}
%\vspace{-0.5cm}
\end{table}

\subsubsection{Case Study}

Finally, we evaluated our DTA mining approach on a concrete scheduling scenario. As claimed in \cite{DBLP:books/sp/22/TangG022}, most task models are subclasses of task automata \cite{FERSMAN20071149Task} -- timed automata extended with real-time tasks
triggered by arrival at states. %The task triggered by arrival at states is 
The combined system of a task model and a scheduler with predetermined scheduling strategy can be formally expressed as a timed automata (see \cite{Wang2025TRE} for a detailed example). To validate whether our approach can effectively extract task models by mining the DTA of such combined systems, we designed a combined system consisting of a sporadic 2-task model scheduled by a FIFO scheduler to evaluate whether our approach is capable of extracting the task model by mining the DTA of the combined system. We denote the 2 tasks as $P_1(b_1,w_1,d_1)$, and $P_2(b_2,w_2,d_2)$, where $b_i$ and $w_i$ are the best-case and the worst-case execution times and $d_i$ is the relative deadline for task $P_i$ $(i=1,2)$. Figure \ref{fig:task} illustrates the task automata $\mathcal{A_0}$. 

 The task automaton can trigger the instances (jobs) of $P_1$ and $P_2$ in any order during execution, subject to the minimum interarrival times of $p_1$ and $p_2$, respectively. We configured the parameters as $(b_1,w_1,d_1,p_1)=(0,2,5,5)$ and $(b_2,w_2,d_2,p_2)=(0,5,20,12)$, and generated traces as follow: we modified the guards from $c_i>p_i$ to $p_i-1<c_i<2p_i$ and allowed $\mathcal{A_0}$ to run stochastically. Consequently, each time $P_i$ is triggered, its interarrival time falls within the range $(p_i-1,2p_i)$. If it falls within $(p_i-1,p_i)$, the current trace is marked as negative. 
 
 The FIFO scheduler is allowed to complete a job with a random execution time within $(b_i,w_i+1)$. And if the execution time exceeds $w_i$, the current trace is marked negative. Additionally, if the completion time misses the relative deadline, the trace will also be marked negative and the trace generation terminates immediately. Traces that are never marked negative are marked positive traces. 
 Additionally, each time task $P_i$ is triggered or completed, an event $a_i$ or $b_i$ is sent, and the time interval between the previous event and the current event is recorded as the current delay time. The maximum number of transitions and guard upperbound are set to $E=1$ and $\kappa=20$, respectively. The dataset comprises 600 positive and 600 negative traces with length between 4 and 10, and the DTA mined from these traces $\mathcal{A}$ is presented in Figure \ref{fig:TA}.

\begin{figure}[t]
    \centering
    \resizebox{1\linewidth}{!}{
    \begin{tikzpicture}[->, >=stealth', shorten >=1pt, auto, node distance=1.5cm, thick,scale=1.5, every node/.style={scale=3}, initial where=above]
    \node[initial, state, accepting] (0) {\makecell[c]{\large $l_1$\\ $P_1(b_1,w_1,d_1)$}};
    \node[state] (1) [right = 3cm of 0] {\makecell[c]{\large $l_2$\\ $P_2(b_2,w_2,d_2)$}};

    \path (0) edge [loop left] node [left] {\large $c_1>p_1,\{c_1\}$} (0)
    (1) edge  [loop right]  node[right] {\large $c_2>p_2,\{c_2\}$} (1)
    (0) edge  [in=135, out=45]  node[ above] {\large $c_2>p_2,\{c_2\}$} (1)
    (1) edge  [in=-45, out=-135] node[below ] {\large $c_1>p_1,\{c_1\}$} (0);
    \end{tikzpicture}
    }
    \caption{The sporadic 2-task automaton $\mathcal{A_0}$.}
    \label{fig:task}
    \end{figure}
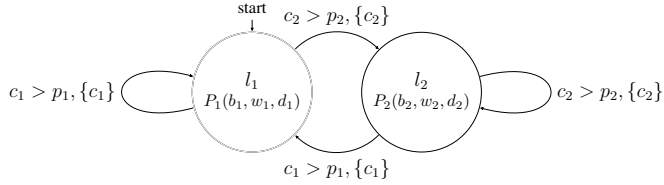

The resulting DTA $\mathcal{A}$ comprises $|Q|=5$ states and $|\mathcal{C}|=3$ clocks. Semantically, each state represents a distinct scheduler state of jobs in queue, while the clocks monitor the interarrival times of $P_1$, $P_2$, and the execution time of the latest job respectively. The learned guards of $\mathcal{A}$ correctly capture that $P_1$ (or $P_2$) cannot be triggered within $p_1=5$ (or $p_2=12$) time units from its last trigger, and the jobs have to be completed within $w_1=2$ (or $w_2=5$) time units. The relative deadlines $d_i$ were not covered as the parameter settings ensure no deadlines are missed. $\mathcal{A}$ is more complex than the models mined in the scalability tests. We attribute that such a more complex DTA is synthesized because the negative examples generated in this case are structurally very similar to positive ones: they share identical event sequences, and the temporal distance between negative and positive traces can be minimal. The probability of such traces appearing is very low in randomly sampled traces, as seen in the scalability tests.

\begin{figure}[t]
    \centering
    \resizebox{0.8\linewidth}{!}{
    \begin{tikzpicture}[->, >=stealth', shorten >=1pt, auto, node distance=1.5cm, semithick,scale=1.5, every node/.style={scale=1}, initial where=above]
    \node[initial, state] (0) {\makecell[c]{\large $q_0$\\ $idle$}};
    \node[state] (1) [below left = 2cm of 0] {\makecell[c]{\large $q_1$\\ $(1)$}};
    \node[state] (2) [below right = 2cm of 0] {\makecell[c]{\large $q_2$\\ $(2)$}};
    \node[state] (3) [below= 1.5cm of 1] {\makecell[c]{\large $q_3$\\ $(1,2)$}};
    \node[state] (4) [below= 1.5cm of 2] {\makecell[c]{\large $q_4$\\ $(2,1)$}};
    \node[rectangle, minimum width=3cm, minimum height=2cm] (box) at (0.2,-1.9) {\large $b_1,y<2,\{y\}$};
    \node[rectangle, minimum width=3cm, minimum height=2cm] (box) at (0.4,-3.8) {\large $b_2,y<5,\{y\}$};
    \path 
    (0) edge  [in=90, out=-150]  node[ above left] {\large $a_1,5<x_1<11,\{x_1,y\}$} (1)
    (1) edge  [in=-120, out=30] node[above ] {\large $b_1,y<2,\emptyset$} (0) 
    (0) edge  [in=90, out=-30]  node[ above right] {\large $a_2,x_2>12,\{x_2,y\}$} (2)
    (2) edge  [in=-60, out=150] node[below] {\large $b_2,y<5,\emptyset$} (0)
    (3) edge  [in=-135, out=45]  node[ above right] {} (2)
    (4) edge  [in=-45, out=180]  node[ below] {} (1)
    (1) edge  [in=90, out=-90]  node[ left] {\large $a_2,x_2>12,\{x_2\}$} (3)
    (2) edge [in=90, out=-90] node[right ] {\large $a_1,5<x_1<10,\{x_1\}$} (4) ;
    \end{tikzpicture}
    }
    \caption{The DTA $\mathcal{A}$ mined from traces.}
    \vspace{-0.5cm}
    \label{fig:TA}
    \end{figure}
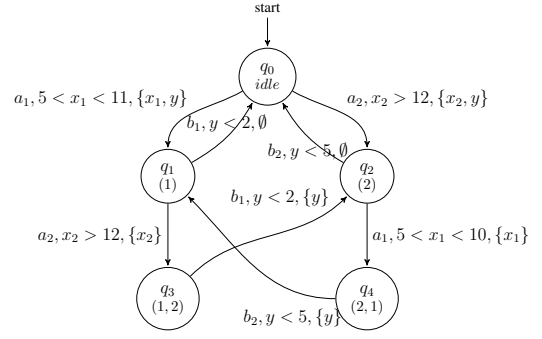

%% file: conclusion.tex
\section{Conclusion} \label{sc:conclusion}
We addressed the passive DTA mining problem by proposing an approach exploiting sEL and tAPTA to simplify the input examples of positive and negative traces, followed by the SMT encoding scheme to synthesize a DTA of specified size from the simplified examples in tAPTA structure. Our implementation, evaluated against existing SMT-based approaches that lack example simplification, demonstrates that our approach significantly reduces the complexity of the SMT formulas encoded. Furthermore, the scalability of the proposed approach was evaluated through experiments on uniformly sampled traces from target DTAs of various sizes, as well as simulated output traces from a scheduling system. 
%We addressed the TRE synthesis problem by theoretically establishing its decidability and proposing a practical framework aimed at finding a minimal TRE consistent with the given set of real-time behaviors. The framework is based on the insight of first finding a $p$TRE that accepts all positive examples and then appropriately deciding its time intervals. We have implemented the synthesis method, along with three strategies for enumerating and pruning $p$TREs. The three strategies were compared experimentally and the framework's scalability was demonstrated through tests on uniformly sampled real-time behaviors of TAs, which are from TREs of various complexities and a train model.

While the SMT-based passive learning framework guarantees a minimal DTA recognizing the directly encoded languages, it cannot ensure minimality for Problem \ref{prob:mDTAmining} when over-approximation is applied. In future work, we plan to investigate refined sample over-approximation methods to restore the proof of minimality. Additionally, we intend to extend our approach to handle nondeterministic situations, which requires a distinct encoding scheme.
%\textbf{Future work.} We plan to explore other enumeration and pruning strategies for our framework. The challenge lies in striking a balance between reducing the effort required to check $p$TREs and expanding the scope of pruning effectively. Additionally, we intend to identify the sufficient conditions for positive (and negative) examples required by our synthesis framework. 
%If positive examples are sufficient, then it must due to the insufficiency of the provided negative examples that the target TRE isn't synthesized, and vice versa. 
%Furthermore, we are to provide a PAC bound regarding the number of samples needed to achieve sufficient conditions of examples incorporating the uniform sampling technique.

\section*{Acknowledgment}
We thank the anonymous reviewers for their valuable comments, which significantly improved the quality of this paper. This work was supported by the National Key R\&D Program of China under grants No.~2022YFA1005100, No.~2022YFA1005101 and No.~2022YFA1005103, the Natural Science Foundation of China (NSFC) under grants No.~W2511064 and No.~62192732, and the CAS Project for Young Scientists in Basic Research (No.\ YSBR-123), the ISCAS Basic Research Grant (No.\ ISCAS-JCZD-202406).

%% file: appendix.tex
\appendices
\section{Proofs of Auxiliary Results}\label{app:proof}
\subsection{Proof of Lemma \ref{lemma:sEL}}

\begin{proof}
Suppose $\omega_1 = (\sigma_1,t_1^1)(\sigma_2,t_2^1)\cdots(\sigma_n,t_n^1)\in L(\mathcal{A})$, $\omega_2 = (\sigma_1,t_1^2)(\sigma_2,t_2^2)\cdots (\sigma_n,t_n^2)$.

$(\Rightarrow)$ For TA $\mathcal{A}=(Q,q_0, \mathcal{C},\Sigma,\delta, F)$, a accepting run of $\mathcal{A}$ that $\omega_1$ is associated with is $run = s_0,\tttrans_1,s_1,\tttrans_2,\cdots,\tttrans_n, s_n$. We denote $Q(s)=q$ and $N(s)=\nu$ for $s=(q,\nu)$. Suppose $k\leq n$ is the largest integer such that there exists $\omega'\in \mathbf{pre}(\omega_2)$ and a run $run' = s_0',\tttrans_1',s_1',\tttrans_2',\cdots,\tttrans_k', s_k',\, s_0'=s_0,\, \forall i \in [1,k] , \, Q(s_i)=Q(s_i')$, $\omega'$ is associated with $run'$. 

If $k<n$, then $Q(s_k)=Q(s_k')$.  Considering $\tttrans_{k+1}$ and $s_{k+1}$, there exists $(Q(s_k),\\ \sigma_{k+1},g,r, Q(s_{k+1}))\in\delta$ such that $\nu_1 =N(s_k)\oplus t_{k+1}^1\models g$. Let $\nu_2 =N(s_k')\oplus t_{k+1}^2$. It can be proved by induction with Def. \ref{def:tasem} that $\forall c\in \mathcal{C}$, $\nu_1(c)=T_{j_c,k+1}(\omega_1), \,\nu_2(c)=T_{j_c,k+1}(\omega_2), \,1\leq j_c\leq k+1$. By Def \ref{def:sel}, $\forall c\in \mathcal{C},$ either $ \nu_1(c)=\nu_2(c)=d$ or $d<\nu_1(c),\nu_2(c)<d+1$. $g$ is the conjunction of integer-bounded inequality, so $\nu_2\models g$. Then there exists $s_{k+1}'=(N(s_{k+1},\nu'))\in S, \, \forall c \in r, \nu'(c)=0, \forall c \in \mathcal{C}\,\backslash\, r, \nu'(c)=\nu_2(c)$ such that $s_k'{\xrightarrow{(\sigma_{k+1},t_{k+1}^2)}}s_{k+1}'$. $\omega''=w'\cdot (\sigma_{k+1},t_{k+1}^2))\in\mathbf{pre}(\omega_2)$ is associated with run $run''=s_0',\tttrans_1',s_1',\tttrans_2',\cdots,\tttrans_k', s_k',\tttrans_{k+1}', s_{k+1}'$, which yields a contradiction.

If $k=n$, then $run'$ is an accepting run since $Q(s_n)=Q(s_k')$ is an accepting state. And $\omega_2$ is the timed word associated with $run'$, so $\omega_2\in L(\mathcal{A})$. 

$(\Leftarrow)$ For a certain pair $(j,k),1\leq j\leq k\leq n$,  suppose a TA $\mathcal{A}_{j,k}=(Q,q_0, \{c\},\Sigma,  \bigcup_{1\leq i\leq n}\{\sigma_i\}, \delta_{j,k}, F)$ as follow:  $Q=\{q_0,q_1,\cdots,q_n\}$,  $F=\{q_n\}$. There is only 1 clock $c$. $\delta_{j,k}= \{(q_{i-1},\sigma_i),g_i,r_i,q_i\, |\, 1\leq i\leq n\}$, where $r_i=\emptyset$ if $i\neq j$, $ r_j={c}$ if $i=j$, $g_i=\top$ if $i\neq k$. If $T_{j,k}(\omega_1)$ is an integer, $g_k= (c\geq  T_{j,k}(\omega_1))\, \wedge\, (c\leq T_{j,k}(\omega_1))$, otherwise $g_k= (c> \lfloor T_{j,k}(\omega_1) \rfloor)\, \wedge\, (c < \lceil T_{j,k}(\omega_1) \rceil)$. By construction $\omega_1\in L(\mathcal{A}_{j,k})$, so $\omega_2\in L(\mathcal{A}_{j,k})$. Then $(T_{j,k}(\omega_2) =  T_{j,k}(\omega_1))$ if $T_{j,k}(\omega_1)$ is an integer, otherwise $\lfloor T_{j,k}(\omega_1) \rfloor< T_{j,k}(\omega_2) < \lceil T_{j,k}(\omega_1) \rceil$. By Def. \ref{def:sel}, $\mathcal{SE}({\omega}_1)= \mathcal{SE}({\omega}_2)$. 
\end{proof}

%\subsection{Proof of Corollary \ref{cor:conflictfree}}
%\begin{proof}
%Given $\Omega=(\Omega_{+},\Omega_{-})$,

%$(\Rightarrow)$ if there exists any $\omega_1\in\Omega_+$, $\omega_1'\in\Omega_-$ such that $\mathcal{SE}({\omega}_1)= \mathcal{SE}({\omega}_1')$, there is a conflict by Proposition \ref{prop:sEL}, which is a contradiction.

%$(\Leftarrow)$ if $\forall \omega_1\in\Omega_+$, $\omega_1'\in\Omega_-$ $\mathcal{SE}({\omega}_1)\neq \mathcal{SE}({\omega}_1')$, then there exist a trivial solution DTA to Problem \ref{prob:DTAmining}. Consequently, $\mathit{A}(\Omega_{+},\Omega_{-}),\mathit{D}(\Omega_{+},\Omega_{-})\neq\emptyset$ and $\Omega$ is conflict-free. A trivial DTA can be inferred from $\Omega$ by first applying the preprocessing and constructing the tAPTA as proposed in Section \ref{sc:problems}, then constructing the DTA from the unsimplified tAPTA as in the process presented in the proof of Lemma \ref{lemma:DTA}.
%\end{proof}

\subsection{Proof of Lemma \ref{lemma:ExMerge}}
\begin{proof}
For a single operation of extended merging, if it is:
\begin{itemize}
    \item States that don't have both positive and negative descendants are merged. There are obviously no conflicts to be introduced.    
%    \item States that have both positive and negative descendants (including themselves) are merged, and the conflict is introduced  
    \item States that have both positive and negative descendants are merged. Suppose before merging the two states are $q_1$ and $q_2$, and after merging the merged state is $q$. $q_+\in F$ and  $q_-\in R$ are a pair of common accepting and rejecting descendants of $q$ and $q_i, i=1,2$. Before merging, there is no conflict in the tAPTA. If any conflict between $L(q_+)$ and $L(q_-)$ is introduced by merging, then there exists $p_+\in path(q_+)$ and $p_-\in path(q_-)$, such that $L(q_+)\cap L(q_-)\neq \emptyset$. 
    If $p_+$ (or $p_-$) doesn't contain a newly merged transition, then before merging there exist paths $p_+'\in path(q_+),L(p'_+)=L(p_+)$ (or $p_-'\in path(q_-),L(p'_-)=L(p_-)$). 
    If only $p_+$ (or $p_-$) contains a newly merged transition that goes from $q$ to $q''$, a common child of $q_1$ and $q_2$ before merging, and the newly merged transition is labeled by $ (\sigma,(I_1,I_2,\cdots,I_n)\vee (I_1',I_2',\cdots,I_m'))$ and $n\neq m$. Then only one of the two branches can be taken $L(p'_+)$ (or $L(p'_-)$). Suppose the branch is $(I_1,I_2,\cdots,I_n)$. Then by replacing the newly merged transition with the transition labeled with $(\sigma,(I_1,I_2,\cdots,I_n))$ that goes to $q''$, and replacing $q$ with the corresponding parent state of $q''$ before merging, we yield a path $p_+'\in path(q_+), L(p'_+)=L(p_+)$ (or $p_-'\in path(q_-),L(p'_-)=L(p_-)$) in the tAPTA before merging. If $n = m$, we can likewise find two paths $p_+',p_+''\in path(q_+), L(p'_+)\cup L(p''_+)\supset L(p_+)$ (or $p_-',p_-''\in path(q_-), L(p'_-)\cup L(p''_-)\supset L(p_-)$) in the tAPTA before merging. No matter in which case, there exists conflict between $L(p'_+)$ (or $L(p'_+)$ and $L(p''_+)$) and $L(p'_-)$ (or $L(p'_-)$ and $L(p''_-)$), which is a contradiction. Hence no conflict is introduced by merging.
    %there exists paths $p_+'\in path(q_+)$ (or $p_-'\in path(q_-)$) such that the only difference between $p_+ $ and $p_+'$ (or between $p_- $ and $p_-'$) is the newly merged in By Def. \ref{def:ExMerge}, 
    %before merging there exist paths $p_+'\in path(q_+)$ and $p_-'\in path(q_-)$, such that the only difference between $p_+ $ and $p_+'$ (or between $p_- $ and $p_-'$) is that $q_i, i=1,2$ in $p_+'$ (or $p_-'$) and the next transition are replaced by $q$ and the merged transition labeled by $ (\sigma,(I_1,I_2,\cdots,I_n)\vee (I_1',I_2',\cdots,I_m'))$. If $n\neq m$, then the branch 
\end{itemize}
\end{proof}

%% file: reference.bib
@article{Alur94,
  author    = {Rajeev Alur and
               David L. Dill},
  title     = {A Theory of Timed Automata},
  _journal   = {Theoretical Computer Science},
  journal   = {Theor. Comput. Sci.},
  volume    = {126},
  number    = {2},
  pages     = {183--235},
  year      = {1994},
  _url       = {https://doi.org/10.1016/0304-3975(94)90010-8},
  _doi       = {10.1016/0304-3975(94)90010-8},
  _note={\doi{10.1016/0304-3975(94)90010-8}}
}

@article{Angluin78,
  author       = {Dana Angluin},
  title        = {On the Complexity of Minimum Inference of Regular Sets},
  journal      = {Inf. Control.},
  volume       = {39},
  number       = {3},
  pages        = {337--350},
  year         = {1978},
  _url          = {https://doi.org/10.1016/S0019-9958(78)90683-6},
  _doi          = {10.1016/S0019-9958(78)90683-6}
}

@article{Grinchtein10,
  author    = {Olga Grinchtein and
               Bengt Jonsson and
               Martin Leucker},
  title     = {Learning of event-recording automata},
  journal   = {Theoretical Computer Science},
  volume    = {411},
  number    = {47},
  pages     = {4029--4054},
  year      = {2010},
  _url       = {https://doi.org/10.1016/j.tcs.2010.07.008},
  _doi       = {10.1016/j.tcs.2010.07.008}
}

@inproceedings{HenryJM20,
  author    = {L{\'{e}}o Henry and
               Thierry J{\'{e}}ron and
               Nicolas Markey},
  _editor    = {Nathalie Bertrand and
               Nils Jansen},
  title     = {Active Learning of Timed Automata with Unobservable Resets},
  _booktitle = {Proceedings of the 18th International Conference on Formal Modeling and Analysis of Timed Systems, {FORMATS} 2020},
  booktitle = {{FORMATS} 2020},
  series    = {LNCS},
  volume    = {12288},
  pages     = {144--160},
  publisher = {Springer},
  year      = {2020},
  _url       = {https://doi.org/10.1007/978-3-030-57628-8\_9},
  _doi       = {10.1007/978-3-030-57628-8\_9}
}

@article{VerwerWW12,
  author       = {Sicco Verwer and
                  Mathijs de Weerdt and
                  Cees Witteveen},
  title        = {Efficiently identifying deterministic real-time automata from labeled
                  data},
  journal      = {Mach. Learn.},
  volume       = {86},
  number       = {3},
  pages        = {295--333},
  year         = {2012},
  _url          = {https://doi.org/10.1007/s10994-011-5265-4},
  _doi          = {10.1007/S10994-011-5265-4}
}

@article{SchmidtGHK13,
  author       = {Jana Schmidt and
                  Asghar Ghorbani and
                  Andreas Hapfelmeier and
                  Stefan Kramer},
  title        = {Learning probabilistic real-time automata from multi-attribute event
                  logs},
  journal      = {Intell. Data Anal.},
  volume       = {17},
  number       = {1},
  pages        = {93--123},
  year         = {2013},
  _url          = {https://doi.org/10.3233/IDA-120569},
  _doi          = {10.3233/IDA-120569}
}

@inproceedings{CornanguerLRT22,
  author       = {L{\'{e}}na{\"{\i}}g Cornanguer and
                  Christine Largou{\"{e}}t and
                  Laurence Roz{\'{e}} and
                  Alexandre Termier},
  title        = {{TAG:} Learning Timed Automata from Logs},
  booktitle    = {Thirty-Sixth {AAAI} Conference on Artificial Intelligence, {AAAI}
                  2022, Thirty-Fourth Conference on Innovative Applications of Artificial
                  Intelligence, {IAAI} 2022, The Twelveth Symposium on Educational Advances
                  in Artificial Intelligence, {EAAI} 2022 Virtual Event, February 22
                  - March 1, 2022},
  pages        = {3949--3958},
  publisher    = {{AAAI} Press},
  year         = {2022},
  _url          = {https://doi.org/10.1609/aaai.v36i4.20311},
  _doi          = {10.1609/AAAI.V36I4.20311}
}

@article{AnWZZZ21,
  author       = {Jie An and
                  Lingtai Wang and
                  Bohua Zhan and
                  Naijun Zhan and
                  Miaomiao Zhang},
  title        = {Learning real-time automata},
  journal      = {Sci. China Inf. Sci.},
  volume       = {64},
  number       = {9},
  year         = {2021},
  _url          = {https://doi.org/10.1007/s11432-019-2767-4},
  _doi          = {10.1007/S11432-019-2767-4}
}

@article{AnZZZ21,
  author       = {Jie An and
                  Bohua Zhan and
                  Naijun Zhan and
                  Miaomiao Zhang},
  title        = {Learning Nondeterministic Real-Time Automata},
  journal      = {{ACM} Trans. Embed. Comput. Syst.},
  volume       = {20},
  number       = {5s},
  pages        = {99:1--99:26},
  year         = {2021},
  _url          = {https://doi.org/10.1145/3477030},
  _doi          = {10.1145/3477030}
}

@article{VerwerWW11,
author       = {Sicco Verwer and Mathijs de Weerdt and Cees Witteveen},
title        = {The efficiency of identifying timed automata and the power of clocks},
journal      = {Inf. Comput.},
volume       = {209},
number       = {3},
pages        = {606--625},
year         = {2011},
_url         = {https://doi.org/10.1016/j.ic.2010.11.023},
_doi         = {10.1016/J.IC.2010.11.023}
}

@inproceedings{AnCZZZ20,
  author       = {Jie An and
                  Mingshuai Chen and
                  Bohua Zhan and
                  Naijun Zhan and
                  Miaomiao Zhang},
  _editor       = {Armin Biere and
                  David Parker},
  title        = {Learning One-Clock Timed Automata},
  _booktitle    = {Tools and Algorithms for the Construction and Analysis of Systems
                  - 26th International Conference, {TACAS} 2020, Held as Part of the
                  European Joint Conferences on Theory and Practice of Software, {ETAPS}
                  2020, Dublin, Ireland, April 25-30, 2020, Proceedings, Part {I}},
  booktitle   = {{TACAS} 2020},
  series       = {Lecture Notes in Computer Science},
  volume       = {12078},
  pages        = {444--462},
  publisher    = {Springer},
  year         = {2020},
  _url          = {https://doi.org/10.1007/978-3-030-45190-5\_25},
  _doi          = {10.1007/978-3-030-45190-5\_25}
}

@inproceedings{XuAZ22,
  author       = {Runqing Xu and
                  Jie An and
                  Bohua Zhan},
  _editor       = {Ahmed Bouajjani and
                  Luk{\'{a}}s Hol{\'{\i}}k and
                  Zhilin Wu},
  title        = {Active Learning of One-Clock Timed Automata Using Constraint Solving},
  booktitle    = {Automated Technology for Verification and Analysis - 20th International
                  Symposium, {ATVA} 2022, Virtual Event, October 25-28, 2022, Proceedings},
  booktitle    = {{ATVA} 2022},
  series       = {Lecture Notes in Computer Science},
  volume       = {13505},
  pages        = {249--265},
  publisher    = {Springer},
  year         = {2022},
  _url          = {https://doi.org/10.1007/978-3-031-19992-9\_16},
  _doi          = {10.1007/978-3-031-19992-9\_16}
}

@inproceedings{DierlHKKLLM23,
author       = {Simon Dierl and Falk Maria Howar and Sean Kauffman and Martin Kristjansen and    Kim Guldstrand Larsen and Florian Lorber and Malte Mauritz},
_editor      = {Kristin Yvonne Rozier and Swarat Chaudhuri},
title        = {Learning Symbolic Timed Models from Concrete Timed Data},
_booktitle   = {{NASA} Formal Methods - 15th International Symposium, {NFM} 2023,
              Houston, TX, USA, May 16-18, 2023, Proceedings},
booktitle    = {{NFM} 2023},
series      = {Lecture Notes in Computer Science},
volume      = {13903},
pages        = {104--121},
publisher    = {Springer},
year         = {2023},
_url         = {https://doi.org/10.1007/978-3-031-33170-1\_7},
_doi         = {10.1007/978-3-031-33170-1\_7}
}

@inproceedings{VaandragerB021,
  author    = {Frits W. Vaandrager and
               Roderick Bloem and
               Masoud Ebrahimi},
  _editor    = {Alberto Leporati and
               Carlos Mart{\'{\i}}n{-}Vide and
               Dana Shapira and
               Claudio Zandron},
  title     = {Learning Mealy Machines with One Timer},
  _booktitle = {Proceedings of the 15th International Conference on Language and Automata Theory and Applications, {LATA} 2021},
  booktitle = {{LATA} 2021},
  series    = {LNCS},
  volume    = {12638},
  pages     = {157--170},
  publisher = {Springer},
  year      = {2021},
  _url       = {https://doi.org/10.1007/978-3-030-68195-1\_13},
  _doi       = {10.1007/978-3-030-68195-1\_13}
}

@inproceedings{TapplerALL19,
author       = {Martin Tappler and Bernhard K. Aichernig and Kim Guldstrand Larsen and Florian Lorber},
_editor      = {{\'{E}}tienne Andr{\'{e}} and Mari{\"{e}}lle Stoelinga},
title        = {Time to Learn - Learning Timed Automata from Tests},
_booktitle   = {Formal Modeling and Analysis of Timed Systems - 17th International
              Conference, {FORMATS} 2019, Amsterdam, The Netherlands, August 27-29,
              2019, Proceedings},
booktitle    = {{FORMATS} 2019},
series      = {LNCS},
volume      = {11750},
pages        = {216--235},
publisher    = {Springer},
year         = {2019},
_url         = {https://doi.org/10.1007/978-3-030-29662-9\_13},
_doi         = {10.1007/978-3-030-29662-9\_13}
}

@inproceedings{AichernigPT20,
author       = {Bernhard K. Aichernig and Andrea Pferscher and Martin Tappler},
_editor      = {Ritchie Lee and Susmit Jha and Anastasia Mavridou},
title        = {From Passive to Active: Learning Timed Automata Efficiently},
_booktitle   = {{NASA} Formal Methods - 12th International Symposium, {NFM} 2020, Moffett Field, CA, USA, May 11-15, 2020, Proceedings},
booktitle    = {{NFM} 2020},
series      = {LNCS},
volume      = {12229},
pages        = {1--19},
publisher    = {Springer},
year         = {2020},
_url         = {https://doi.org/10.1007/978-3-030-55754-6\_1},
_doi         = {10.1007/978-3-030-55754-6\_1}
}

@inproceedings{TapplerAL22,
author       = {Martin Tappler and Bernhard K. Aichernig and Florian Lorber},
_editor      = {Jyotirmoy V. Deshmukh and Klaus Havelund and Ivan Perez},
title        = {Timed Automata Learning via {SMT} Solving},
_booktitle   = {{NASA} Formal Methods - 14th International Symposium, {NFM} 2022,
              Pasadena, CA, USA, May 24-27, 2022, Proceedings},
booktitle    = {{NFM} 2022},
series      = {LNCS},
volume      = {13260},
pages        = {489--507},
publisher    = {Springer},
year         = {2022},
_url         = {https://doi.org/10.1007/978-3-031-06773-0\_26},
_doi         = {10.1007/978-3-031-06773-0\_26}
}

@inproceedings{Waga23,
  author       = {Masaki Waga},
  _editor       = {Constantin Enea and
                  Akash Lal},
  title        = {Active Learning of Deterministic Timed Automata with Myhill-Nerode
                  Style Characterization},
  _booktitle    = {Computer Aided Verification - 35th International Conference, {CAV}
                  2023, Paris, France, July 17-22, 2023, Proceedings, Part {I}},
  booktitle    = {{CAV} 2023},
  series       = {LNCS},
  volume       = {13964},
  pages        = {3--26},
  publisher    = {Springer},
  year         = {2023},
  _url          = {https://doi.org/10.1007/978-3-031-37706-8\_1},
  _doi          = {10.1007/978-3-031-37706-8\_1}
}

@article{TengZA24,
      title={Learning Deterministic Multi-Clock Timed Automata}, 
      author={Yu Teng and Miaomiao Zhang and Jie An},
      year={2024},
      _eprint={2404.07823},
      _archivePrefix={arXiv},
      journal   = {arXiv,},
      volume    = {abs/2404.07823},
      note = {To appear in HSCC 2024.},
      primaryClass={cs.FL}
}

@InProceedings{Francois1998APTA,
author="Coste, Fran{\c{c}}ois
and Nicolas, Jacques",
editor="Honavar, Vasant
and Slutzki, Giora",
title="How considering incompatible state mergings may reduce the DFA induction search tree",
booktitle="Grammatical Inference",
year="1998",
publisher="Springer Berlin Heidelberg",
address="Berlin, Heidelberg",
pages="199--210",
isbn="978-3-540-68707-8"
}

@InProceedings{10.1007/978-3-031-71177-0_4,
author="Dell'Erba, Daniele
and Li, Yong
and Schewe, Sven",
editor="Platzer, Andre
and Rozier, Kristin Yvonne
and Pradella, Matteo
and Rossi, Matteo",
title="DFAMiner: Mining Minimal Separating DFAs from Labelled Samples",
booktitle="Formal Methods",
year="2025",
publisher="Springer Nature Switzerland",
address="Cham",
pages="48--66",
isbn="978-3-031-71177-0"
}

@article{10.1162/089120100561601,
    author = {Daciuk, Jan and Mihov, Stoyan and Watson, Bruce W. and Watson, Richard E.},
    title = {Incremental Construction of Minimal Acyclic Finite-State Automata},
    journal = {Computational Linguistics},
    volume = {26},
    number = {1},
    pages = {3-16},
    year = {2000},
    month = {03},
    issn = {0891-2017},
    doi = {10.1162/089120100561601},
    url = {https://doi.org/10.1162/089120100561601},
    eprint = {https://direct.mit.edu/coli/article-pdf/26/1/3/1797487/089120100561601.pdf},
}

@InProceedings{Meng20263DRTA,
author="Meng, Junjie
and An, Jie
and Li, Yong
and Turrini, Andrea
and Zhang, Miaomiao",
editor="Chen, Yu-Fang
and Jensen, Thomas
and Leng{\'a}l, Ond{\v{r}}ej",
title="SAT-Based Synthesis of Minimal Deterministic Real-Time Automata via 3DRTA Representation",
booktitle="Verification, Model Checking, and Abstract Interpretation",
year="2026",
publisher="Springer Nature Switzerland",
address="Cham",
pages="173--196",
isbn="978-3-032-15700-3"
}

@inproceedings{Barbot2023WORDGEN,
  author       = {Beno{\^{\i}}t Barbot and
                  Nicolas Basset and
                  Alexandre Donz{\'{e}}},
  title        = {Wordgen : a Timed word Generation Tool},
  booktitle    = {Proceedings of the 26th {ACM} International Conference on Hybrid Systems: Computation and Control, {HSCC} 2023},
  pages        = {16:1--16:7},
  publisher    = {{ACM}},
  year         = {2023},
  _url          = {https://doi.org/10.1145/3575870.3587116},
  _doi          = {10.1145/3575870.3587116}
}

@INPROCEEDINGS{Wang2025TRE,
  author={Wang, Ziran and An, Jie and Zhan, Naijun and Zhang, Miaomiao and Zhang, Zhenya},
  booktitle={2025 IEEE Real-Time Systems Symposium (RTSS)}, 
  title={On Synthesis of Timed Regular Expressions}, 
  year={2025},
  volume={},
  number={},
  pages={311-323},
  doi={10.1109/RTSS66672.2025.00033}}

@incollection{DBLP:books/sp/22/TangG022,
  author       = {Yue Tang and
                  Nan Guan and
                  Wang Yi},
  editor       = {Yu{-}Chu Tian and
                  David C. Levy},
  title        = {Real-Time Task Models},
  booktitle    = {Handbook of Real-Time Computing},
  pages        = {469--487},
  publisher    = {Springer},
  year         = {2022},
  _url          = {https://doi.org/10.1007/978-981-287-251-7\_29},
  _doi          = {10.1007/978-981-287-251-7\_29}
}

@article{FERSMAN20071149task,
title = {Task automata: Schedulability, decidability and undecidability},
journal = {Information and Computation},
volume = {205},
number = {8},
pages = {1149-1172},
year = {2007},
issn = {0890-5401},
doi = {https://doi.org/10.1016/j.ic.2007.01.009},
url = {https://www.sciencedirect.com/science/article/pii/S0890540107000089},
author = {Elena Fersman and Pavel Krcal and Paul Pettersson and Wang Yi}
}

@inproceedings{MouraB08,
  author       = {Leonardo Mendon{\c{c}}a de Moura and
                  Nikolaj S. Bj{\o}rner},
  _editor       = {C. R. Ramakrishnan and
                  Jakob Rehof},
  title        = {{Z3:} An Efficient {SMT} Solver},
  booktitle    = {14th International Conference on Tools and Algorithms for the Construction and Analysis of Systems, {TACAS} 2008},
  series       = {LNCS},
  volume       = {4963},
  pages        = {337--340},
  publisher    = {Springer},
  year         = {2008},
  _url          = {https://doi.org/10.1007/978-3-540-78800-3\_24},
  _doi          = {10.1007/978-3-540-78800-3\_24}
}

@article{Franck2021Veri,
  author = {Franck  Cassez;Peter Gjøl  Jensen;Kim  Guldstrand Larsen},
  title = {Verification and Parameter Synthesis for Real-Time Programs using Refinement of Trace Abstraction *},
  journal = {Fundamenta Informaticae},
  volume = {178},
  number = {1-2},
  pages = {31-57},
  year = {2021},
  doi = {10.3233/FI-2021-1997},
  URL = {https://doi.org/10.3233/FI-2021-1997},
  eprint = {https://doi.org/10.3233/FI-2021-1997}
}

@InProceedings{Bouyer2005Model,
author="Bouyer, Patricia
and Laroussinie, Fran{\c{c}}ois
and Reynier, Pierre-Alain",
editor="Pettersson, Paul
and Yi, Wang",
title="Diagonal Constraints in Timed Automata: Forward Analysis of Timed Systems",
booktitle="Formal Modeling and Analysis of Timed Systems",
year="2005",
publisher="Springer Berlin Heidelberg",
address="Berlin, Heidelberg",
pages="112--126",
isbn="978-3-540-31616-9"
}

@InProceedings{Jiang2021Safe,
author="Jiang, Zhihao
and Pajic, Miroslav
and Moarref, Salar
and Alur, Rajeev
and Mangharam, Rahul",
editor="Flanagan, Cormac
and K{\"o}nig, Barbara",
title="Modeling and Verification of a Dual Chamber Implantable Pacemaker",
booktitle="Tools and Algorithms for the Construction and Analysis of Systems",
year="2012",
publisher="Springer Berlin Heidelberg",
address="Berlin, Heidelberg",
pages="188--203",
isbn="978-3-642-28756-5"
}

@INPROCEEDINGS{Ouaknine2005MTL,
  author={Ouaknine, J. and Worrell, J.},
  booktitle={20th Annual IEEE Symposium on Logic in Computer Science (LICS' 05)}, 
  title={On the decidability of metric temporal logic}, 
  year={2005},
  volume={},
  number={},
  pages={188-197},
  doi={10.1109/LICS.2005.33}}

@InProceedings{Behrmann2005Schedule,
author="Behrmann, Gerd
and Larsen, Kim G.
and Rasmussen, Jacob I.",
editor="de Boer, Frank S.
and Bonsangue, Marcello M.
and Graf, Susanne
and de Roever, Willem-Paul",
title="Priced Timed Automata: Algorithms and Applications",
booktitle="Formal Methods for Components and Objects",
year="2005",
publisher="Springer Berlin Heidelberg",
address="Berlin, Heidelberg",
pages="162--182",
isbn="978-3-540-31939-9"
}

@inproceedings{WallnerALT25,
  author       = {Felix Wallner and
                  Bernhard K. Aichernig and
                  Florian Lorber and
                  Martin Tappler},
  title        = {Mutating Skeletons: Learning Timed Automata via Domain Knowledge},
  booktitle    = {{IEEE} International Conference on Software Testing, Verification
                  and Validation, {ICST} 2025 - Workshops, Naples, Italy, March 31 -
                  April 4, 2025},
  pages        = {67--77},
  publisher    = {{IEEE}},
  year         = {2025},
  url          = {https://doi.org/10.1109/ICSTW64639.2025.10962513},
  doi          = {10.1109/ICSTW64639.2025.10962513},
  bibsource    = {dblp computer science bibliography, https://dblp.org}
}

@inproceedings{ChenSZLM23,
  author       = {Hanyue Chen and
                  Yu Su and
                  Miaomiao Zhang and
                  Zhiming Liu and
                  Junri Mi},
  editor       = {Constantin Enea and
                  Akash Lal},
  title        = {Learning Assumptions for Compositional Verification of Timed Automata},
  booktitle    = {Computer Aided Verification - 35th International Conference, {CAV}
                  2023, Paris, France, July 17-22, 2023, Proceedings, Part {I}},
  series       = {Lecture Notes in Computer Science},
  volume       = {13964},
  pages        = {40--61},
  publisher    = {Springer},
  year         = {2023},
  url          = {https://doi.org/10.1007/978-3-031-37706-8\_3},
  doi          = {10.1007/978-3-031-37706-8\_3},
  bibsource    = {dblp computer science bibliography, https://dblp.org}
}
